\documentclass[11pt,letterpaper]{article}
\usepackage{enumitem}
\usepackage{amsmath}
\usepackage{amsthm}
\usepackage{amssymb}
\usepackage{mathtools}
\usepackage{dsfont}
\usepackage[dvipsnames]{xcolor}
\usepackage{bbm,color,tikz}

\newcommand{\comment}[1]{}
\newcommand{\Omit}[1]{}

\newcommand{\uri}[1]{{\color{blue}{UF: #1}}}

\newtheorem{theorem}{Theorem}[section]

\newtheorem{lemma}[theorem]{Lemma}
\newtheorem{proposition}[theorem]{Proposition}
\newtheorem{claim}[theorem]{Claim}
\newtheorem{definition}[theorem]{Definition}
\newtheorem{example}[theorem]{Example}
\newtheorem{observation}[theorem]{Observation}

\newcommand{\vect}[1]{\ensuremath{\mathbf{#1}}}
\newcommand{\prices}{\vect{p}}

\newcommand{\pricej}{p_{j}}

\newcommand{\reals}{\mathbb{R}}

\renewcommand{\emptyset}{\varnothing}

\begin{document}

\title{Are Gross Substitutes a Substitute for Submodular Valuations?}


\author{
Shahar Dobzinski\thanks{Weizmann Institute and Microsoft Research Israel; {\tt dobzin@gmail.com }}
\and
Uriel Feige\thanks{Weizmann Institite and Microsoft Research Israel; {\tt Uriel.Feige@weizmann.ac.il}}
\and
Michal Feldman\thanks{Tel-Aviv University and Microsoft Research Israel; {\tt mfeldman@tau.ac.il}}
\and
Renato Paes Leme\thanks{Google Research NYC; {\tt renatoppl@google.com}}
}

\maketitle

\begin{abstract}
    The class of gross substitutes (GS) set functions plays a central role in Economics and Computer Science.
	GS belongs to the hierarchy of {\em complement free} valuations introduced by Lehmann, Lehmann and Nisan, along with other prominent classes: $GS \subsetneq Submodular \subsetneq XOS \subsetneq Subadditive$.
	The GS class has always been more enigmatic than its counterpart classes, both in its definition and in its relation to the other classes. 
	For example, while it is well understood how closely the Submodular, XOS and Subadditive classes (point-wise) approximate one another, approximability of these classes by GS remained wide open.	
	
	Our main result is the existence of a submodular valuation (one that is also budget additive) that cannot be approximated by GS within a ratio better than $\Omega(\frac{\log m}{\log\log m})$, where $m$ is the number of items. 
	En route, we uncover a new symmetrization operation that preserves GS, which may be of independent interest.
	
	We show that our main result is tight with respect to budget additive valuations. Additionally,	for a class of submodular functions  that we refer to as {\em concave of Rado} valuations (this class greatly extends budget additive valuations), we show approximability by GS within an $O(\log^2m)$ factor.

\end{abstract}

\section{Introduction}
\label{sec:introduction}
A {\em valuation function} over a set $M$ of $m$ items is a function $f: 2^M \rightarrow \reals$ that assigns a real value $f(S)$ to every set $S \subseteq M$, which is additionally {\em monotone} ($f(T) \leq f(S)$ for every $T \subseteq S$) and {\em normalized} ($f(\emptyset)=0$). Valuation functions are commonly used to describe combinatorial preferences over items. In many settings it is natural to assume that the valuations additionally belong to a specific class, e.g., they are submodular or subadditive. Indeed, algorithms and impossibilities for different classes of valuation functions were developed for numerous problems, e.g., welfare maximization \cite{dobzinski2010approximation, feige2009maximizing, vondrak2008optimal}, truthful mechanisms \cite{assadi2019improved,dobzinski2007two}, price of anarchy \cite{bhawalkar2011welfare, feldman2013simultaneous}, and learning \cite{balcan2011learning, badanidiyuru2012sketching, feldman2013representation}, to name a few examples. 

In Algorithmic Game Theory, the work that has set the tone for the study of valuations functions is that of Lehmann, Lehmann and Nisan~\cite{lehmann2006combinatorial}. They presented a hierarchy of ``complement-free" valuation functions, whose five most prominent classes (listed from the least general to the most general) are OXS, Gross Substitutes (GS), Submodular, XOS, and Subadditive. A typical line of work attempts to determine the performance of algorithms, for example, in terms of their approximation ratios, in each of the classes of the hierarchy.

Among the five classes of the hierarchy, two of them were defined semantically, as the set of all valuations that satisfy a natural property like submodularity and subadditivity. Two of them (OXS and XOS) were defined syntactically, as the set of all valuations that can be constructed by applying certain OR- and XOR-like operations. The remaining class of GS valuations stands out in the sense that its definition is neither syntactic nor semantic but -- to some extent -- coincidental. It was essentially defined by Kelso and Crawford \cite{kelso1982job} as a condition on the valuation functions that is needed for a certain auction to end in an equilibrium.

Thus, it may not come as a total surprise that the class of gross substitutes valuations remained as perhaps the least understood class among the classes considered by \cite{lehmann2006combinatorial}. The class is lacking properties that ``natural'' classes of valuations tend to possess, e.g., it is not closed under addition and applying a concave function on a GS valuation does not necessarily result in a GS valuation. Despite that, the GS class plays a central role in many settings (for example, it is, in some formal sense, the largest class for which a Walrasian equilibrium exists \cite{gul1999walrasian}, it guarantees exact welfare maximization in polynomial time \cite{nisan2006communication}, VCG outcomes are guaranteed to be in the core \cite{ausubel2002ascending}). One consequence to our partial understanding of the class is a lack of ``complicated'' GS valuations or a lack of good techniques to construct them, although such valuations are crucial for proving hardness of various tasks\footnote{\label{foot:1}One exception that proves the rule is the matroid rank function construction of \cite{balcan2011learning} that proves the non-existence of good sketches for GS valuations.}. For example, easy-to-construct XOS valuations often serve as the hardest instances even for subadditive valuations (e.g., \cite{dobzinski2010approximation,badanidiyuru2012sketching}). Recent attempts \cite{ostrovsky2015gross,BalkanskiL18} to constructively produce the set of all GS valuations had only partial success.

We suggest another direction toward deciphering the enigmatic character of the GS class. Instead of focusing on exact characterizations of the class, we suggest to study its ``closeness'' to other classes of valuations. We rely on the notion of \emph{approximation} to explore the proximity of the class of GS valuations to other classes.

\begin{definition}
A valuation function $g$ approximates a valuation function $f$ within a ratio $\rho \ge 1$ if for every set $S$ of items we have that $g(S) \le f(S) \le \rho \cdot g(S)$.
A class $C_1$ of valuation functions approximates a class $C_2$ of valuation functions within a ratio $\rho \ge 1$ if for every function $f \in C_2$ there is a function $g \in C_1$ such that $g$ approximates $f$ within the ratio $\rho$.
\end{definition}

The fact that a class $C_1$ approximates $C_2$ by some small factor may not have immediate consequences in all settings of interest. However, more often than not, a small approximation factor does hint that algorithms for $C_1$ will work well on valuations from $C_2$, perhaps with some modifications, and that impossibility results for $C_2$ also apply to $C_1$.

For example, consider the complement-free hierarchy (see also Figure \ref{fig:relations}): $GS \subsetneq Submodular \subsetneq XOS \subsetneq Subadditive$. It is known that XOS approximates Subadditive within a ratio of $O(\log m)$ and no better~\cite{dobzinski2007two,bhawalkar2011welfare} (recall that $m$ denotes the number of items). Moreover, the approximability gap of Submodular and XOS is $\Omega(\sqrt m)$~\cite{badanidiyuru2012sketching} and this is tight (up to poly-log factors)~\cite{badanidiyuru2012sketching, goemans2009approximating}. Indeed, the large gap between Submodular and XOS is evident when considering, e.g., that value queries are very effective for various optimization tasks with submodular valuations (maximization subject to cardinality constraint \cite{nemhauser1978analysis}, welfare maximization \cite{vondrak2008optimal}, minimization \cite{iwata2001combinatorial}) but provide only poor approximation ratios under XOS valuations. Similarly, the relatively small gap between XOS and Subadditive may explain why in many settings the best algorithms for subadditive valuations achieve comparable results to the best algorithms for XOS valuations \cite{feige2009maximizing,assadi2021improved}.

\begin{figure}
	\centering
	\includegraphics[width=0.5\textwidth]{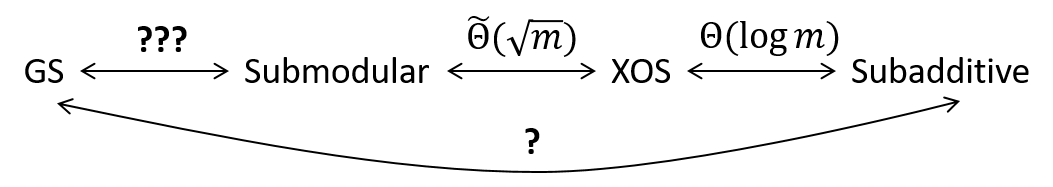}
	\caption{Previous approximability results among complement-free valuations.
	\label{fig:relations}}
\end{figure}

As evident from Figure~\ref{fig:relations}, the approximation relationship of GS and the other classes was unknown before this paper. The first question is to determine the exact approximation ratio of Subadditive by GS. Previous work implies that it is between $\sqrt m$ (since Submodular functions cannot approximate XOS to within a better factor) and $m$ (every subadditive function $v$ can be trivially approximated within a factor of $m$ by the additive function in which the value of item $j$ is $\frac{1}{m}v(\{j\})$). The second question is to determine the approximation factor of Submodular by GS \footnote{The literature contains examples of submodular functions that are not GS. These examples use a constant number of items and imply that there is a constant approximation gap between GS and Submodular (see a 3-item example in \cite{lehmann2006combinatorial}).}

The first question turns out to be much easier than the second one. Previous results \cite{badanidiyuru2012sketching} already tell us that subadditive functions can be approximated to within a factor of $\tilde O(\sqrt m)$ by applying a concave function on an additive valuation. We prove -- as part of a more general result -- that any concave function of an additive valuation can be approximated by a GS function to within a logarithmic factor. This establishes that GS can approximate subadditive valuations to within a factor of $\tilde O(\sqrt m)$.

Answering the second question is more challenging and serves as the main focus of this paper. 

\vspace{0.1in}
\begin{theorem}[Main Result]
\label{thm:MainThm}
The class of Gross Substitutes (GS) valuation functions does not approximate the class of Submodular valuation functions within a ratio better than $\Omega(\frac{\log m}{\log\log m})$.
\end{theorem}
\vspace{0.1in}

To prove this theorem we present a simple-to-describe family of submodular valuations (parametrized by integers $m$ and $d$) that we refer to as $BA(k,d)$, and show that members of this family cannot be approximated well by any GS function. $BA(k,d)$ even belongs to the class of budget additive valuations\footnote{A valuation is budget additive if there exists some $B$ such that for every bundle $S$, $v(S)=\min(\Sigma_{j\in S}v(\{j\}),B)$.}.

The proof of Theorem~\ref{thm:MainThm} makes extensive use of symmetries. We say that items $i$ and $j$ are symmetric in a valuation $f$ if for every bundle $S \subset M \setminus \{i,j\}$ it holds that $f(S \cup \{i\})=f(S \cup\{j\})$. We show that symmetries have far reaching implications on GS functions. Specifically, if $f$ is GS, then every two items that are symmetric are also {\em weak substitutes} of each other, in the sense that the marginal value of one item given the other cannot be larger than the marginal value of the item given any third item. We refer to functions with this property as {\em symmetric weak substitutes}, SWS. (See Definition~\ref{def:weakSubstitutes} and Proposition~\ref{pro:GSisSWS}.) Using the SWS property it is trivial to exhibit submodular functions that are not GS. 
For example, consider the budget additive function $f$ defined on three items $a,b,c$, with $f(a) = f(b) = \frac{1}{2}$, $f(c) = 1$, and $f(S) = \min[1, \sum_{i\in S} f(i)]$. Items $a$ and $b$ are symmetric, but they are not weak substitutes, as $f(a|c) < f(a|b)$. Hence $f$ is not GS. The functions $BA(k,k)$ (for $k \geq 2$) used in the proof of Theorem~\ref{thm:MainThm} are generalizations of this function $f$. 

To derive inapproximability ratios using the SWS framework, we partition the proof into two parts. One is to show that the approximation ratio of any GS function that preserves the symmetries of $BA(k,d)$ is $\Omega(\frac{\log m}{\log\log m})$ (for appropriate choices of $k$ and $d$). For this we extend the SWS property to groups of items (see Lemma~\ref{lem:key}). With this extension, and the careful design of the functions $BA(k,d)$, proving this inapproximability result is relatively straightforward. 

The second part of the proof shows that for every function $f$ (and thus also $BA(k,d)$), among GS functions, the one that approximates $f$ best can be assumed to preserve the symmetries of $f$. Typically, such statements are proved by using symmetrization operations.  Symmetrization is usually relatively easy to obtain in ``natural'' classes like submodular and subadditive classes, where a symmetrized version of a valuation can be obtained by appropriately permuting the items and averaging. However, here we face the following difficulty: 
the average of two GS functions is not necessarily GS. Hence, we introduce a different symmetrization operation that preserves the approximation ratio. Interestingly, applying our new symmetrization operation to a submodular function does not guarantee that the symmetrized function is submodular, but if the original function is GS then the symmetrized function is GS as well.



We do not know whether the $\frac{\log m}{\log\log m}$ factor is tight, that is, whether submodular valuations can be approximated by GS to within a factor of $O(\frac{\log m}{\log\log m})$. However, we do show that every budget additive valuation can be approximated by a GS function within this factor. In fact, we prove a more general result, which is described next.

\paragraph{Approximating Concave of GS by GS.}

Observe that every budget additive valuation $f$ is a composition $f = g \circ f'$ of a concave monotone function $g: \reals^+ \rightarrow \reals^+$ (where $g(x) = x$ for $x \le B$, and $g = B$ for $x \ge B$) with an additive valuation function $f'$. Note that the corresponding $f'$, being additive, belongs to GS. 

In general, we can consider functions of the form $f = g \circ f'$, where $g$ is an arbitrary concave monotone function, and $f'$ is an arbitrary GS valuation function. Note that such a function $f$ is necessarily submodular, as $f'$ is submodular, and submodularity (unlike the property of being GS) is preserved under composition with concave functions. We show that GS-ness is approximately preserved under composition of concave functions, at least for some GS valuations.

Rado valuations can be viewed as {\em weighted matroid matching} valuations --- a strict superclass of OXS valuations (corresponding to weighted matching in bipartite graphs, see definition in Section~\ref{sec:preliminaries}). Rado valuations are known to be GS, and until very recently it has been open whether they contain the entire GS class; it was recently shown to be false \cite{GargHV21}.

\vspace{0.1in}
\begin{theorem}
\label{thmWMM}
Let $g: \reals^+ \rightarrow \reals^+$ be an arbitrary concave monotone function, and let $f'$ be an arbitrary {\em Rado} 
valuation function (over $m$ items). 
The function $f = g \circ f'$ can be approximated within a ratio of $O(\log^2 m)$ by a gross substitutes valuation function. 
As a special case (with stronger results), if $f'$ is additive, or more generally, a weighted Matroid Rank Function, then the approximation ratio is $O(\log m)$.
\end{theorem}

\paragraph{Open Questions.} The most obvious -- and important -- question that the paper leaves open is to determine the exact ratio that GS valuations can approximate submodular valuations. Another immediate question is to determine whether the class of all GS valuations is approximately closed under concave functions; i.e.,  whether the GS class approximates the class of concave of GS within a poly-logarithmic factor.

{
More generally, it would be interesting to study how well GS functions can approximate and be approximated by other valuation functions. 
Three subclasses of submodular valuations that are of interest are budget additive (BA), coverage, and matroid rank sum (MRS --- sum over matroid rank functions, see Definition~\ref{def:mrs}). MRS is a strict superclass of coverage functions (coverage is the class of matroid rank functions of rank 1, see Definition~\ref{def:mrs}); all three valuations are incomparable to GS. 

It is not hard to see that BA and coverage functions approximate GS very badly ($\Omega(\sqrt{m})$). Indeed, BA has a polynomial representation, and every coverage function can be approximated by a coverage function with polynomial representation {\cite{badanidiyuru2012sketching}}. If either family approximates GS well, then GS would have good sketches, contradicting footnote~\ref{foot:1}. 
How well can MRS approximate GS (or submodular) valuations is an open problem.
In the other direction, how well can GS valuations approximate coverage and MRS valuations? {We show (see Theorem~\ref{th:coverage-ba}) that coverage functions approximate BA functions within constant factors, and hence our negative results concerning inapproximability of BA by GS within a better ratio than $\Omega(\frac{\log m}{\log \log m})$ extend also to approximating coverage or MRS by GS.}
}

\section{Preliminaries}
\label{sec:preliminaries}

In what follows, we define some valuation functions relevant to the current paper. Additional valuations, including matroid rank functions (MRF), weighted matroid rank function (WMRF), matroid rank sum (MRS), and OXS are defined in Appendix~\ref{sec:matroids-prelim}.
For all functions $f$ below, $f:2^{M} \rightarrow \reals^+$ (For simplicity of presentation, we sometimes omit the parenthesis of sets when referring to singletons).
We denote by $f(j \mid X)$ the {\em marginal value} of $j$ given $X$; i.e., $f(j \mid X) = f(X \cup \{j\}) - f(X)$.

\begin{itemize}[noitemsep]
	\item $f$ is {\em additive} is there exist $v_1, \ldots, v_n \in \reals^+$ such that for every $S \subseteq M$, $f(S)=\sum_{j \in S}v_j$.
	\item $f$ is {\em budget additive} (BA) if there exist $v_1, \ldots, v_n, B \in \reals$ such that for every $S \subseteq M$, $f(S)=\min\{B,\sum_{j \in S}v_j\}$.
	\item $f$ is {\em coverage} if there exists a finite set $\Omega$, where every element $j \in \Omega$ is associated with a weight $w_j \in \reals^+$, and a function $g:M \rightarrow 2^{\Omega}$ such that for every set $S \subseteq M$, $f(S)=\sum_{j \in \bigcup_{i\in S}g(i)}w_j$.
	\item $f$ is Rado if there is a bipartite graph $G(M,V;E)$ with non-negative weights on the edges, and a matroid $\mathcal{M}=(V,\mathcal{I})$, such that for every $S\subseteq M$, $f(S)$ is the total weight of the max weight matching on the subgraph induced by $S$ and a subset of $V$ that belongs to $\mathcal{I}$.	
	\item $f$ is {\em submodular} if for every sets $S,T \subseteq M$ such that $S \subseteq T$, and item $j$, $f(j \mid T) \leq f(j \mid S)$.
	\item $f$ is {\em XOS} (also known as fractionally subadditive) if it is a maximum over additive functions; i.e., there exist additive functions $f_1, \ldots, f_k$ for some integer $k$, such that for every $S \subseteq M$, $f(S)=\max_{i\in [k]}\{f_i(S)\}$.
\end{itemize}

We next turn to the gross substitutes (GS) class, whose definition uses the notions of prices, utility and demand.
Suppose every item $j \in M$ is associated with a price $p_j$.
Then, the {\em utility} derived from a set $S$ of items, given valuation function $f$ and item prices $\prices = (p_1, \ldots, p_m)$, is the net benefit from $S$; i.e., $u(S,p) = f(S)-\sum_{j\in S}\pricej$.
Consequently, given a valuation function $f$ and item prices $\prices$, the {\em demand} is the collection of sets of items that maximize utility.

\begin{definition}
	Given a valuation function $f$, and a price vector $\prices = (p_1, \ldots, p_m) \in \reals^{m}$, the {\em demand} of $f$ under $\prices$ is $D(\prices) = argmax_{S\subseteq M}\{f(S)-\sum_{j\in S}\pricej\}$.
\end{definition}

We are now ready to define gross substitutes valuations.
\begin{definition}\label{def:gs_kc}
	A valuation function $f$ is gross substitutes (GS) if for every pair of price vectors $\prices, \vect{q}$ such that $\prices \leq \vect{q}$, for every set $S \in D(\prices)$, there exists a set $T \in D(\vect{q})$ such that $T$ contains every item $j \in S$ such that $p_j = q_j$.
\end{definition}
Conceptually, this definition says that if $S$ is the demand under a given price vector, and we increase the prices of some items, then those items whose prices have not increased are still in demand.
%

The following containment relations are known: 
$$
OXS,WMRF \subsetneq Rado \subsetneq GS \subsetneq Submodular \subsetneq XOS \subsetneq Subadditive.
$$
In addition, it is known that $BA,Coverage,MRS \subsetneq submodular$, and every pair among BA, Coverage, GS and MRS is non-comparable, except for $Coverage \subsetneq MRS$.

\vspace{-0.1in}
\paragraph{GS characterization}

The gross substitute (GS) class has several alternative characterizations.
In what follows we give two characterizations that are of particular interest in this paper.

A {\em demand query} gets a price vector $\prices$ as input and returns a set $S \in D(\prices)$. The first characterization uses the notion of {\em marginal utility}. Given a set $S \subseteq M$ and an item $j$, the marginal utility of $j$ with respect to $S$ given prices $\prices$ is $u(S\cup\{j\},\prices)-u(S,\prices)$.


\begin{lemma}(The greedy characterization of GS)\cite{DT1995,PaesLeme17}
	A valuation $f$ is GS if and only if every demand query can be answered by the following greedy algorithm: at every step pick an item $j$ of highest marginal utility, until no item has positive marginal utility.
	\label{lem:greedy-char}
\end{lemma}
The next characterization does not use the notion of prices. {For two sets $T,S \subset M$, we use $f(T|S)$ to denote the marginal value of $T$ with respect to $S$; i.e., $f(T|S)=f(T \cup S)-f(S)$. For simplicity of notation we omit the curly brackets of sets when clear in the context.}

\begin{lemma}(The triplet characterization of GS)\cite{ReijniersePG2002}
	A function $f$ is GS if and only if it is submodular, and for every set $S$, and items $i,j,k\not\in S$, it holds that
	\begin{equation}
	\label{eq:no-unique-max}
	\max\{f(i,k|S)+f(j|S), f(j,k|S)+f(i|S)\} \geq f(i,j|S)+f(k|S).
	\end{equation}
	We refer to Equation~\eqref{eq:no-unique-max} as the {\em triplet condition}.
	\label{lem:triplet-char}
\end{lemma}

The following lemma (whose proof is deferred to Appendix~\ref{sec:app-prelim}) shows that the triplet characterization of GS holds in an approximate sense for all submodular {valuation} functions.


\begin{lemma}
	\label{lem:local-approx}
	Let $f$ be a submodular {valuation} function. Then, for every set $S$, and items $i,j,k\not\in S$, 
	$$
	\max\{f(i,k|S) + f(j|S), f(j,k|S) + f(i|S)\} \ge \frac{3}{4}(f(i,j|S) + f(k|S)).
	$$
	The ratio of $\frac{3}{4}$ is best possible.
\end{lemma}

Thus, the triplet condition of GS is a local condition that approximately (within a constant factor) holds for all submodular valuations.
In contrast, our main result limits the extent to which approximability holds globally, as it shows the existence of submodular valuation functions that cannot be approximated by a GS function within any constant factor.\footnote{
	Let us remark that if $f$ is XOS, then we might have that $f(i,j|S) + f(k|S) = 1$ but $\max\{f(i,k|S) + f(j|S), f(j,k|S) + f(i|S)\} = 0$.
	For example, if $S$ contains a single item $s$, and $f(T) = \max\{|T \cap \{s\}|, |T \cap \{i,j\}|\}$ for every set $T$.}

{Finally, we recall the definition of the convolution (max welfare) operation.

\begin{definition}
\label{def:conv}
Given two valuation functions $f,g$, the convolution of $f$ and $g$, denoted $f \star g$, is
$$
(f\star g)(S)=\max_{T \subseteq S}\left(f(T)+g(S \setminus T)\right)
$$
\end{definition}

The gross substitutes property is closed under convolution 
 \cite{lehmann2006combinatorial,Murota96}
}

\section{Lower Bounds via Symmetrization}
\label{sec:log-gap}




Throughout this section we denote the set $X\cup \{j\}$ by $Xj$ for simplicity of presentation.

\begin{definition}
	\label{def:symmetric}
	Let $f$ be an arbitrary set function.
	For two items $i,j \in M$, we say that $f$ is $(i,j)$-symmetric,
	denoted by $i =_f j$, if for every set $S \subset M \setminus \{i,j\}$ it holds that $f(Si) = f(Sj)$.
\end{definition}

\begin{proposition}
	\label{pro:equivalence}
	The relation $=_f$ is an equivalence relation.
\end{proposition}

\begin{proof}
	By definition, $=_f$ is reflexive ($i =_f i$) and symmetric ($i =_f j$ iff $j =_f i$). It remains to show that it is transitive, namely, that $i =_f j$ and $j =_f k$ imply $i =_f k$. There are two cases to consider. In one $j \not\in S$, and then $f(Si) = f(Sk)$ follows from $f(Si) = f(Sj) = f(Sk)$. In the other $j \in S$, and then $S$ can be written as $S'j$ for $S' = S \setminus \{j\}$. Then $f(S'ji) = f(S'jk)$ follows from $f(S'ji) = f(S'ij) = f(S'ik) = f(S'ki) = f(S'kj) = f(S'jk)$. 
\end{proof}

If follows that $=_f$ partitions $M$ into equivalence classes, where every two items are symmetric with respect to $f$ if and only if they are in the same equivalence class. We refer to these equivalence classes as {\em symmetry classes}.

Given $M$, we say that a partition $P_1$ is a {\em coarsening} of a partition $P_2$ (and $P_2$ is a {\em refinement} of $P_1$), if every class of $P_1$ is a union of (one or more) classes of $P_2$. Equivalently, every class of $P_2$ is contained in (or equal to) a class of $P_1$. (A partition is both a coarsening and a refinement of itself.)

The following theorem plays a key role in establishing large gaps between various functions and the class of GS functions.

\begin{theorem}
	\label{thm:symmetrization}
	Let $f$ be a valuation function and, for some $\rho \ge 1$, let $g$ be a GS function that approximates $f$ within a ratio of $\rho$. Then there is a GS function $g_{sym}$ that approximates $f$ within a ratio not worse than $\rho$, and moreover, the symmetry classes of $g_{sym}$ form a coarsening of the symmetry classes of $f$.
\end{theorem}

The standard approach towards proving a theorem such as Theorem~\ref{thm:symmetrization} is to construct additional versions of $g$ that differ from the original $g$ only in the sense that item names are permuted (using permutations that respect the symmetries of $f$), and then $g_{sym}$ is taken to be the average of all versions of $g$. This gives a function that respects the symmetries of $f$, and moreover, preserves the approximation ratio of $g$. However, this symmetrization technique does not apply to GS functions, because the average of two GS functions is not necessarily a GS function. For example, consider a BA function with budget 2 and item values 2,2,0. This function is GS. Applying the above symmetrization (with respect to the last two items) implies averaging it with a BA function with budget 2 and item values 2,0,2. The average of the two functions gives a BA function with budget 2 and values 2,1,1, and this function is not GS.

Despite the above, our proof of Theorem~\ref{thm:symmetrization} is based on
repeatedly applying to $g$ two-item symmetrization operations until we reach the desired function $g_{sym}$. Given a function $g$ and two items $i$ and $j$ such that $g$ is not $(i,j)$-symmetric, an $(i,j)$-symmetrization produces a new function $g_1$ that is $(i,j)$-symmetric. Necessarily, our two-item symmetrization procedure is based on an operation that is different from averaging. We refer to our operation as two-item max-symmetrization.

\begin{definition}
	\label{def:max-symmetrization}
	Given a set function $g$ and two items $i,j \in M$, {\em two-item max-symmetrization} with respect to $i$ and $j$ produces the function $g_1$ that for every set $S \subset M \setminus \{i,j\}$ satisfies:
	
	\begin{itemize}
		
		\item $g_1(S) = g(S)$.
		
		\item $g_1(Sij) = g(Sij)$.
		
		\item $g_1(Si) = g_1(Sj) = \max[g(Si), g(Sj)]$.
		
	\end{itemize}
	
\end{definition}


{If two-item max-symmetrization is applied to a submodular function, the resulting function need not be submodular, as shown by the following example.} 

\begin{example}
	\label{ex:notSubmodular}
	For the following monotone submodular function $g$ on items $\{a,b,c,d\}$, two-item max-symmetrization with respect to items $a$ and $b$ produces a function $g_1$ that is not submodular. 
	$g(\emptyset)=0,
	g(a)=3,
	g(b)=2,
	g(c)=1,
	g(d)=2,
	g(ab)=5,
	g(ac)=3,
	g(ad)=3,
	g(bc)=3,
	g(bd)=4,
	g(cd)=3,
	g(abc)=5,
	g(abd)=5,
	g(acd)=3,
	g(bcd)=5,
	g(abcd)=5
	$.
	One can verify that $g$ is submodular, but $g_1(c|a)=0$ while $g_1(c|da)=1$, violating submodularity.
\end{example}

{In contrast to Example~\ref{ex:notSubmodular}, we have the following lemma for GS functions.}

\begin{lemma}
	\label{lem:twoItemSymmetrization}
	If two-item max-symmetrization is applied to a GS function $g$, the resulting function $g_1$ is also GS.
\end{lemma}

\begin{proof}
Let $g$ be a GS function, and let $x$ and $y$ be the two items with respect to which we apply the two-item max-symmetrization operation, producing the function $g_1$. We use the convention that $S$ denotes an arbitrary set that does not contain neither $x$ nor $y$.








We now introduce an operation that goes half way towards two-item max-symmetrization. It will be called {\em partial 2-symmetrization}, and for it, the order between $x$ and $y$ matters. That is, {\em partial 2-symmetrization} gives (possibly) different results when applied with $(x,y)$ or when applied with $(y,x)$. We specify the outcome function $q$ when partial $(x,y)$-symmetrization is applied to $g$.

\begin{itemize}

\item $q(S) = g(S)$.

\item $q(Sx) = g(Sx)$

\item $q(Sy) = \max[g(Sx),g(Sy)] = g_1(Sx)$.

\item $q(Sxy) = g(Sxy)$.

\end{itemize}

When partial $(y,x)$-symmetrization is applied to $g$, the roles of $x$ and $y$ are interchanged, and then $q(Sy) = g(Sy)$ and $q(Sx) = \max[g(Sx),g(Sy)]$.

\begin{claim}
\label{claim:partial}
For every GS function $g$ and for every $x,y$, the partial $(x,y)$-symmetrization of $g$ (denoted by $q$) is a GS function.
\end{claim}

\begin{proof}
Add an item $z$ to the set of items. Extend $g$ to the new set of items by giving $z$ marginal value~0. {Function} $g$ remains GS. 

Let $t$ denote $\max[g(x),g(y)]$. Introduce an auxiliary GS function $g'$, where $g'(T) = t$ if $T$ intersects $\{x,y,z\}$, and $g'(T) = 0$ otherwise. {Function} $g'$ is GS, {since it is unit-demand}.

We extend the convention regarding $S$ introduced above so that now $S$ denotes an arbitrary set that does not contain any of $x,y,z$.


{Consider the function $h$ which is the convolution of $g$ and $g'$ (see Definition~\ref{def:conv}). Then $h$ is GS  \cite{lehmann2006combinatorial,Murota96}, and satisfies:}

\begin{itemize}

\item $h(Sx) = g(S) + g'(x) = g(S) + t$. 

\item $h(Sxz) = g(Sx) + g'(z) = g(Sx) + t$.

\item $h(Sxy) = \max[g(Sx),g(Sy)] + t = {g_1}(Sx) + t$. 

\item $h(Sxyz) = g(Sxy) + g'(z) = g(Sxy) + t$.

\end{itemize}

Consider the function $p$ that is $h$ endowed with item $x$ ($p(S) = h(Sx) - h(x)$). Function $p$ is GS and satisfies:

\begin{itemize}

\item $p(S) = h(Sx) - h(x) = g(S) + t - t = g(S)$.

\item $p(Sz) = h(Sxz)- h(x) = g(Sx) + t - t = g(Sx)$.

\item $p(Sy) = h(Sxy) - h(x) = g_1(Sx) + t - t = g_1(Sx)$.

\item $p(Syz) = h(Sxyz) - h(x) = g(Sxy) + t - t = g(Sxy)$.

\end{itemize}

Observe that the left hand side in the equalities above (function $p$) may contain $z$ but never contains $x$, whereas the right hand side (functions $g$ and {$g_1$}) may contain $x$ but never contains $z$. Now, for the function $p$, rename item $z$ to be item $x$. After this renaming, observe that the function $p$ is precisely the function $q$. It follows that $q$ is GS, as desired.
\end{proof}


Armed with Claim~\ref{claim:partial}, we can complete the proof of Lemma~\ref{lem:twoItemSymmetrization}.
Apply first partial $(x,y)$-symmetrization to $g$. By Claim~\ref{claim:partial}, the resulting function $q$ is GS. Now apply partial $(y,x)$-symmetrization to $q$. Observe that the resulting function is $g_1$. By Claim~\ref{claim:partial}, $g_1$ is GS.
\end{proof}


{The max-symmetrization (Definition \ref{def:max-symmetrization}) and partial $(x,y)$-symmetrization operations are special cases of a more general operation called \emph{induction by networks}. This operation was first identified by Kobayashi,  Murota and Tanaka \cite{kobayashi2007operations} in the context of $M$-convex functions on jump systems. Their insight extends to other classes of functions such as valuated matroids (a fact that is explored in Husic et al\cite{husic2021complete} to characterize Rado-valuations) and GS. In Appendix \ref{appendix:induction-by-networks} we provide a definition of this operation for GS as well as a direct proof that it preserves GS (following from Claim \ref{claim:partial}).}

Given Lemma~\ref{lem:twoItemSymmetrization}, we can prove Theorem~\ref{thm:symmetrization}.

\begin{proof}
	Suppose that valuation function $f$ is approximated within a ratio of $\rho \ge 1$ by some GS function $g$. Denote $g$ by $g_0$, and construct a sequence of functions $g_t$ (for $t \ge 0$) as follows. For a given $t \ge 0$, if there is a pair of items $i,j \in M$ such that $f$ is $(i,j)$-symmetric and $g_t$ is not, apply the two-item max-symmetrization procedure (where $i$ and $j$ are the two items) to $g_t$, obtaining $g_{t+1}$.
	
	By Lemma~\ref{lem:twoItemSymmetrization}, all functions $g_t$ are GS.
	
	We now consider the approximation ratio. Initially, for every set $S$ it holds that $g(S) \le f(S) \le \rho g(S)$. Values of sets cannot decrease by two-item max-symmetrization. Consequently, we also have $f(S) \le \rho g_t(S)$ for all $t$. Moreover, the inequality $g_t(S) \le f(S)$ also holds, because when a set increases its value (say, the set $Si$), we have the chain of inequalities $g_t(Si) = \max[g_{t-1}(Si), g_{t-1}(Sj)] \le \max[f(Si), f(Sj)] = f(Si)$, where the left equality is by definition of two-item max-symmetrization, the middle inequality is by induction, and the last equality is because $f$ is $(i,j)$-symmetric.
	
	We next show that the sequence $\{g_t\}$ is finite. Consider the potential function $\psi$ that given a function $g$ sums up all its values. That is, $\psi(g) = \sum_{S \subset M} g(S)$. Observe that two-item max-symmetrization increases the value of at least one set and does not decrease the value of any set, and hence the $\psi$ values of the functions in the sequence $\{g_t\}$ form a strictly increasing sequence. Moreover, starting with $g_0 = g$, the $\psi$ values are supported on only finitely many possible values. This is because for every $t \ge 1$ and every $S \subset M$, there is some $T \subset M$ such that $g_t(S) = g_{t-1}(T)$, and by induction, there also is some $T \subset M$ such that $g_t(S) = g_0(T)$. Hence every $\psi$ value is a sum of $2^m - 1$ values (one value for every nonempty subset of $M$), each taken from a pool of at most $2^m - 1$ values (the distinct values of nonempty sets under $g$).  Consequently, there are less than $2^{2m}$ possible $\psi$ values. Hence the sequence $\{g_t\}$ has length at most $2^{2m}$.
	
	The last function in the sequence $\{g_t\}$ can serve as $g_{sym}$.  For every $i,j \in M$ for which $f$ is $(i,j)$-symmetric, so is $g_{sym}$ (as otherwise $g_{sym}$ would not be the last function in the sequence). For this $g_{sym}$, the partition of $M$ into symmetry classes is a coarsening of the symmetry classes of $f$. Moreover, by being a member of the sequence $\{g_t\}$, the function $g_{sym}$ is GS and approximates $f$ within a ratio of $\rho$.
\end{proof}

Having established Theorem~\ref{thm:symmetrization}, we now prove structural properties that must hold for symmetric GS functions.

Two items $a$ and $b$ are considered to be {\em substitutes} to each other if the marginal value of any one of them, given the other, is~0. We introduce a relaxation of the notion of substitutes, that we refer to as {\em weak substitutes}. Informally, $a$ and $b$ being weak substitutes means that $b$ is the item whose inclusion leads to the most severe decrease in the marginal value of $a$ (compared to including any other item), though unlike the case of substitutes, this decrease does not necessarily reduce the marginal value of $a$ all the way down to~0.

\begin{definition}
	\label{def:weakSubstitutes}
	Given a set function $f$, two items $a$ and $b$ are {\em weak substitutes} if for every set $S\subset M \setminus \{a,b\}$ and item $c \not \in S$, both  $f(a|Sc) \ge f(a|Sb)$ and $f(b|Sc) \ge f(b|Sa)$ hold. We say that a function $f$ is {\em symmetric weak substitutes} (SWS) if every two items that belong to the same symmetry class of $f$ are weak substitutes.
\end{definition}

\begin{proposition}
	\label{pro:GSisSWS}
	Every gross substitutes (GS) function is also symmetric weak substitute (SWS).
\end{proposition}

\begin{proof}
	Let $g$ be a GS function that is $(a,b)$-symmetric. It suffices to prove that $g(a|Sc) \ge g(a|Sb)$ holds for every set $S\subset M \setminus \{a,b\}$ and item $c \not\in S$. (The proof that $g(b|Sc) \ge g(b|Sa)$ is identical.)
	
	The triplet condition (Lemma~\ref{lem:triplet-char}) with the triplet $(a,b,c)$ relative to the set $S$ implies that:
	
	$$g(c|S) + g(ab|S) \le \max[g(a|S) + g(bc|S), g(b|S) + g(ac|S)]$$
	
	As $g$ is $(a,b)$-symmetric, we have that $g(a|S) + g(bc|S) = g(b|S) + g(ac|S)$. It follows that $g(c|S) + g(ab|S) \le g(b|S) + g(ac|S)$. Using $g(ab|S) = g(b|S) + g(a|Sb)$ and $g(ac|S) = g(c|S) + g(a|Sc)$ we have:
	
	$$g(c|S) + g(b|S) + g(a|Sb) \le g(b|S) + g(c|S) + g(a|Sc)$$
	
	Cancelling identical terms from both sides we get the desired $g(a|Sb) \le g(a|Sc)$.
\end{proof}

We have established that every GS function $g$ is also SWS. This implies that pairs of items that are symmetric under $g$ are also weak substitutes. The following lemma extends this property from pairs of items to sets of items, of arbitrary size.


\begin{lemma}
	\label{lem:key}
	Let $g$ be an SWS function (or a GS function, as a special case), and let $B \subset M$ be a set of items that all belong to the same symmetry class in $f$. Partition $B$ into $B_1$ and $B_2$ arbitrarily. Let $C_1$ be an arbitrary set of items, disjoint from $B_2$, and with $|C_1| = |B_1|$. Then $g(B_2|C_1) \ge g(B_2|B_1)$.
\end{lemma}

\begin{proof}
	We may assume without loss of generality that $B_2$ contains only a single item $b_2$. (For example, if $B_2$ contains also an additional item $b$, then $g(B_2|C_1) =  g(b_2| C_1) + g(b| C_1b_2)$. We first add $b_2$, and then apply the lemma again by taking $B'_2 = \{b\}$, $C'_1 = C_1b_2$, and $B'_1 = B_1b_2$.)
	
	Assume for the sake of contradiction that $g(b_2|C_1) < g(b_2|B_1)$. For $0 \le k \le |B_1|$, consider the hybrid sets $H_k$ that contain the first $k$ items of $C_1$, and $|B_1| - k$ items of $B_1$. Then $H_0 = B_1$ and $H_{|B_1|} = C_1$, and by a hybrid argument, there is a value of $k$ for which $g(b_2|H_{k+1}) < g(b_2|H_k)$. Let $c$ be the item in $H_{k+1} \setminus H_k$, let $b_1$ the item in $H_k \setminus H_{k+1}$, and let $H = H_k \cap H_{k+1}$. We get that $g(b_2 | Hc) < g(b_2 | Hb_1)$, contradicting Proposition~\ref{pro:GSisSWS}.
\end{proof}

We remark that SWS is not closed under neither average nor convolution; see Appendix~\ref{app:sws-closed}.

\subsection{Gap of $\Omega(\frac{\log m}{\log \log m})$ between GS and Budget-Additive via Symmetrization}
\label{sec:log-ba}

In this section we prove Theorem~\ref{thm:MainThm} using the symmetrization technique.
By Theorem~\ref{thm:symmetrization}, to show that no GS function approximates a given function $f$ within a ratio better than $\rho$, it suffices to show that no GS function that respects the symmetries of $f$ approximates $f$ within a ratio better than $\rho$.

Consider the following class of budget-additive (BA) functions, parameterized by $k$ and $d$.

\begin{definition}[$BA(k,d)$]
	\label{def:ba-k-d}
	For integers $k \ge 2$ and $d \ge 1$ the budget-additive function denoted by $BA(k,d)$ is defined as follows. Its items are partitioned into $d+1$ levels. Level $h$, for $0 \le h \le d$, has $k^h$ items, each of value $k^{-h}$. The budget is~1.
\end{definition}

	
In Theorem~\ref{th:ba-negativeNew} we establish a lower bound on the ratio by which $BA(k,d)$ can be approximated by a GS function. This lower bound is an expression that depends on the parameters $k$ and $d$. 
The desired lower bound of $(1-o(1))\frac{\log m}{\log\log m}$ is then obtained for $k=d\log d$. 
For the special case of $BA(k,1)$ (i.e., $d=1$) we obtain a lower bound of $\frac{2k}{k+1}$, which is tight (by Proposition~\ref{prop:approx-2}).
This case is presented next as a warmup. 
Recall that $BA(k,1)$ has one item $a$ of value~1, and $k$ items $b_1, \ldots, b_k$, each of value $1/k$, and a budget of 1.

\begin{theorem} (warmup)
	$BA(k,1)$ cannot be approximated by GS within a ratio better than $\frac{2k}{k+1}$.
\end{theorem}

\begin{proof}
	Let $f=BA(k,1)$. 
	We show that no GS function that respects the symmetries of $f$ approximates $f$ within a ratio better than $\frac{2k}{k+1}$. 
	Let $g$ be a GS function that respects the symmetries of $f$.
	Let $B$ denote the set of items $b_1, \ldots, b_k$.
	
	If $g(a)<\frac{k+1}{2k}$, then $\frac{f(a)}{g(a)}>\frac{2k}{k+1}$ and we are done.
	Hence we may assume that $g(a)\geq \frac{k+1}{2k}$. 
	Since all sets have value at most~1, $g(B \mid a) \leq \frac{k-1}{2k}$.
	Let $B_2 = B \setminus \{b_1\}$.
	Observe that $g(B_2 \mid a) \leq g(B \mid a) \leq \frac{k-1}{2k}$. 
	By Lemma~\ref{lem:key}, $g(B_2 \mid b_1) \le g(B_2 \mid a)$.
	In addition, $g(b_1) \le \frac{1}{k}$.
	We get $g(B) = g(b_1) + g(B_2 \mid b_1) \leq \frac{1}{k} + \frac{k-1}{2k} = \frac{k+1}{2k}$, whereas $f(B) = 1$.
\end{proof}

%
%
%
%
%
%
%
%

We now extend this result to $BA(k,d)$.

\begin{theorem}
	\label{th:ba-negativeNew}
	Let $\rho(k,d) > 1$ be the smallest ratio by which $BA(k,d)$ can be approximated by a GS function. Then: 
$$\frac{1}{\rho(k,d)} \le \frac{1}{d+1}\left(1 + \sum_{h=1}^{d} \frac{hk^{h-1}}{k^d}\right) < \frac{1}{d+1} + \frac{1}{k-1}$$ 	
	In particular, for every $\epsilon > 0$, for sufficiently large $d$ and $k = d\log d$, the function $BA(k,d)$ cannot be approximated by GS within a ratio better than
	$(1 - \epsilon)\frac{\log m}{\log\log m}$.
\end{theorem}

\begin{proof}

	Let $f=BA(k,d)$. For convenience, we denote $\frac{1}{\rho(k,d)}$ by $\rho$. Hence  $0 < \rho < 1$ is the largest value (and then $\rho(k,d) = \frac{1}{\rho}$ is the smallest value) such that there is a GS function $g$ that approximates $f$ within a ratio of $\frac{1}{\rho}$. Namely, for every set $S$ we have that $\rho f(S) \le g(S) \le f(S)$. By Theorem~\ref{thm:symmetrization}, we may assume that $g$ respects the symmetries of $f$. We shall derive linear constraints that can be used in order to upper bound the value of $\rho$. For this purpose, we introduce some notation.

\begin{itemize}

\item $L_h$, for $0 \le h \le d$, denotes the set of items in level $h$. Note that $|L_h| = k^h$.

\item $S_h$, for $0 \le h \le d$, denotes the set of items up to (and including) level $h$. Hence $S_h = \bigcup_{i=0}^h L_h$, and $|S_h| = \sum_{i=0}^h k^i$.

\item $L_h^p$ ($p$ stands for prefix) denotes an arbitrary subset of $\sum_{i=0}^{h-1} k^i = |S_{h-1}|$ items from $L_h$. By symmetry, it will not matter for us which items of $L_h$ are in the set $L_h^p$.

\item Given a set $L_h^p$, we denote $L_h \setminus L_h^p$ by $L_h^s$ ($s$ stands for suffix). Observe that $L_h^s$ contains $k^h - \sum_{i = 0}^{h-1} k^i$ items from $L_h$. 

\end{itemize}

For every $1 \le h \le d$,
Lemma~\ref{lem:key} implies that $g(L_h^s | L_h^p) \le g(L_h^s | S_{h-1})$. This is because $|L_h^p| = |S_{h-1}|$. Hence using the notation of Lemma~\ref{lem:key}, $L_h^s$ can serve as $B_2$, $L_h^p$ can serve as $B_1$, and $S_{h-1}$ can serve as $C_1$. 

Using the facts that $L_h = L_h^p \cup L_h^s$ and that $g(L_h^p) \le f(L_h^p) \le \frac{1}{k^h}|L_h^p|$ we then have that $g(L_h) =  g(L_h^p) + g(L_h^s | L_h^p) \le \frac{\sum_{i=0}^{h-1} k^i}{k^h} + g(L_h^s | S_{h-1}) \le \frac{\sum_{i=0}^{h-1} k^i}{k^h} + g(L_h | S_{h-1})$.

Using the above we have that:

\begin{eqnarray*}
\sum_{0 \le h \le d} g(L_h) & \le &  g(L_0) + \sum_{h=1}^d \left(\frac{\sum_{i=0}^{h-1} k^i}{k^h} + g(L_h | S_{h-1})\right)\\ 
& = & \sum_{h=1}^d \frac{\sum_{i=0}^{h-1} k^i}{k^h} + g(L_0) + \sum_{h=1}^d g(L_h | S_{h-1}).
\end{eqnarray*}

Observe that $g(S_d) = g(M) \le f(M) \le 1$. Consequently, $g(L_0) + \sum_{h=1}^d g(L_h | S_{h-1}) \le 1$.
We get
$$
\sum_{0 \le h \le d} g(L_h)
 \le \sum_{h=1}^d \frac{\sum_{i=0}^{h-1} k^i}{k^h} + 1.$$
Consequently, 

$$\rho \le \min_{0 \le h \le d} g(L_h) \le \frac{1}{d+1}\left(1 + \sum_{h=1}^d \frac{\sum_{i=0}^{h-1} k^i}{k^h}\right) = \frac{1}{d+1}\left(1 + \sum_{h=1}^{d} \frac{hk^{h-1}}{k^d}\right)$$


The above upper bound on $\rho$ can be replaced by a simpler upper bound.
Observe that $\sum_{h=1}^{d} hk^{h-1} \le d \sum_{h=1}^{d} k^{h-1} < d\frac{k^d}{k-1}$. Hence 
$$\rho \le \frac{1}{d+1}\left(1 + \sum_{h=1}^{d} \frac{hk^{h-1}}{k^d}\right) < \frac{1}{d+1} + \frac{1}{k-1}$$ 

For $BA(k,d)$ with sufficiently large $d$ and $k = d\log d$, we have that $d = (1+ o(1)) \frac{\log m}{\log\log m}$, where $m = (1 + o(1))k^d$ is the number of items. (Here $o(1)$ denotes terms that tend to~0 as $d$ grows.)  Hence in that case $\rho < \frac{1}{d+1} + \frac{1}{d\log d - 1} < (1 + o(1))\frac{\log\log m}{\log m}$, as desired.
\end{proof}

\noindent {\bf Remark:}
In Appendix~\ref{sec:factor2-sep} we show that the same lower bound of $\frac{2k}{k+1}$ with respect to $BA(k,1)$ can be proved via an approach  that uses the greedy characterization of GS functions, and does not make use of Theorem~\ref{thm:symmetrization}, and that this lower bound is tight. However, we do not know how to extend that proof technique so as to recover the bounds that Theorem~\ref{th:ba-negativeNew} establishes for $BA(k,d)$ with $d > 1$. 

\subsection{Gap of $\sqrt{m}$ between GS and XOS via Symmetrization}

{In~\cite{badanidiyuru2012sketching} it is proved that submodular valuation functions do not approximate XOS valuation functions within a ratio better than $\Omega(\sqrt m)$. As GS functions are a subclass of submodular functions, it follows that GS does not approximate XOS within a ratio better than $\Omega(\sqrt m)$. In this section we use the symmetrization techniques in order to present an alternative proof of this latter fact, and do so for the same subclass of XOS functions that are used in the proof of~\cite{badanidiyuru2012sketching}. For that subclass (when approximated by GS), our proof gives exact tight gaps ($\sqrt{m}$ and not just $\Omega(\sqrt{m})$). 
}


\begin{proposition}
	\label{pro:XOS}
	If $m$ is a perfect square, then XOS cannot be approximated by GS within a ratio better than $\sqrt{m}$.
\end{proposition}

\begin{proof}
	Items are partitioned into $\sqrt{m}$ groups $M_1, \dots, M_{\sqrt m}$, each of size $\sqrt{m}$. The XOS function is $f = \max_i\{f_i\}$, where $f_i$ is additive (with item values~1) over the $i$th group.
	We show that no GS function that respects the symmetries of $f$ approximates $f$ within a ratio better than $\sqrt m$. 
	
	Let $g$ be a GS function that respects the symmetries of $f$. Let $T$ be a set with one item from each group (all such sets $T$ have the same value by symmetry), and observe that $g(T) \le f(T) = 1$.
	Consider a random item $x\in T$ and a random permutation $\pi$ over the items of $T$. The expected (expectation taken both over choice of $x$ and choice of $\pi$) marginal value of $x$ (marginal value compared to the prefix of $\pi$ that precedes $x$) is exactly $\frac{g(T)}{|T|} \le \frac{1}{\sqrt{m}}$. Consequently, there exists an item $x\in T$, for which the above expectation (now taken only over choice of $\pi$) is at most $\frac{1}{\sqrt{m}}$. Consider the group $M_i$ that $x$ belongs to. By Proposition~\ref{pro:GSisSWS}, the marginal contribution of the $j$th item to $g(M_i)$ is not larger than the marginal contribution of $x$ to $g(T)$ when $x$ is in the $j$th location. As there are $\sqrt{m}$ locations and the expected contribution at a random location (corresponding to the random permutation $\pi$) is at most $\frac{1}{\sqrt{m}}$, we have that $g(M_i) \le 1$. But $f(M_i) = \sqrt{m}$, showing that the approximation ratio cannot be better than $\sqrt{m}$.  
\end{proof}


{Unlike GS valuation functions, submodular valuation functions approximate the function $f$ of the proof of Proposition~\ref{pro:XOS} within a ratio somewhat better than $\sqrt{m}$. Let $m = q^2$ for an integer $q \ge 2$. Then the following submodular function $g$ approximates $f$ within a ratio of $q - \frac{1}{2} < \sqrt{m}$. For a set $S$ of items, if $|S| < q$, then $g(S)= \frac{2|S|}{2q-1}$. If $|S| > q$, then $g(S)= \frac{2q}{2q-1}$. If $|S| = q$ then there are two cases. If the $q$ items of $S$ come from distinct groups, then $g(S)= 1$, and if at least two items come from the same group, then $g(S) = \frac{2q}{2q-1}$.} 


\section{Upper Bounds for Concave Functions of (some) GS Functions}
\label{sec:special-cases}

In this section we establish approximability results for concave functions of some GS functions, by GS functions. 
Section~\ref{sec:ba-positive} establishes an upper bound of $(1 + o(1))\frac{\log m}{\log \log m}$ with respect to budget-additive functions. 
This bound is tight due to Theorem~\ref{th:ba-negativeNew}.
Budget additive functions are concave functions of additive functions.  
In Section~\ref{sec:concave} we present a unified approach for establishing upper bounds for concave functions of a more general class of GS valuations. The generality of this approach may come at some loss in the approximation factor (e.g., it gives $\log m$ approximation for BA functions, compared with the $\log m/\log \log m$ given in Section~\ref{sec:ba-positive}).
Finally, in Section~\ref{sec:app-positive} we use the techniques developed in Section~\ref{sec:concave} to prove Theorem~\ref{thmWMM}, establishing that every concave function of a Rado function can be approximated by a GS function within a ratio of $O(\log^{2} m)$.

\subsection{Approximability of Budget-Additive within a Ratio $O(\frac{\log m}{\log \log m})$}
\label{sec:ba-positive}

\begin{theorem}
	\label{thm:ba-loglog}
	Every BA function $f$ can be approximated by a GS function within a ratio of $(1 + o(1))\frac{\log m}{\log \log m}$, where here $o(1)$ denotes a term that tends to~0 as $m$ grows.
\end{theorem}

Before proving Theorem~\ref{thm:ba-loglog}, we introduce a new family of GS functions. We then show that a member of this family serves to prove Theorem~\ref{thm:ba-loglog}.

\newcommand{\minn}[2]{\min\{#1,#2\}}

Let $g: M \rightarrow \reals^{\geq 0}$ be a function that assigns a real value to every item, and let $T: M \rightarrow \reals^{\geq 0}$ be a monotone non-increasing {\em threshold} function, that assigns a real value to every position $1, \ldots, m$.
Given an item $j$, and a position $r$, we define the marginal value of $j$ with respect to position $r$ as 
\begin{equation}
\label{eq:mu}
\mu(j,r) = \minn{g(j)}{T(r)}.
\end{equation}

The function $\mu$ satisfies the following property.

\begin{claim}
	\label{cl:exchange}
	Let $j,k$ be two items such that $g(j) \geq g(k)$. Then, for every position $r$,
	$$
	\mu(j,r)+\mu(k,r+1) \geq \mu(k,r)+\mu(j,r+1).
	$$
\end{claim}

\begin{proof}
	We show that for every $a,b,\alpha,\beta$ such that $a\geq b, \alpha \geq \beta$, 
	\begin{equation}
	\label{eq:exchange}
	\minn{a}{\alpha}+\minn{b}{\beta} \geq \minn{b}{\alpha} + \minn{a}{\beta}.
	\end{equation}
	Indeed, if $b \leq \beta$, then (\ref{eq:exchange}) holds iff $\minn{a}{\alpha} + b \geq b + \minn{a}{\beta}$, which holds by $\alpha \geq \beta$.
	Otherwise, $\beta < b$, and (\ref{eq:exchange}) holds iff $\minn{a}{\alpha}+\beta \geq \minn{b}{\alpha} + \beta$, which holds by $a \geq b$. 
\end{proof}

Given a set of items $S$, sort them in a non-increasing order according to $g$ (breaking ties arbitrarily), and let $S_r$ be the item in position $r$. 
Consider the function $h: 2^{M} \rightarrow \reals^{\geq 0}$ defined as
\begin{equation}
\label{eq:h}
h(S)=\sum_{r=1}^{|S|}\mu(S_r,r).
\end{equation}
I.e., $h$ sums up the marginal values of its items w.r.t. their corresponding positions.


\begin{proposition}
	\label{prop:h-GS}
	For every function $g:M \rightarrow \reals^{\geq 0}$ and non-increasing threshold function $T: [m] \rightarrow \reals^{\geq 0}$, the function $h$ defined in Equation~(\ref{eq:h}) is monotone GS. 
\end{proposition}

{
\begin{proof}
Consider a bipartite graph $G(M,V;E)$ with $|V| = |M|$ and with non-negative weights on the edges, where for every $j \in M$ and $r \in V$, the weight of edge $(j,r)$ is $\mu(j,r)$ (with $\mu$ as in the definition of $h$). Let $h'$ be the set function where for every $S \subset M$, $h'(S)$ is the total weight of the maximum weight matching on the subgraph induced by $S$ and $V$. Then by definition (see Definition~\ref{def:oxs}) $h'$ is an OXS function. We show below that for every set $S\subset M$, $h(S) = h'(S)$, concluding that $h$ is an OXS function, and hence monotone and GS, proving Proposition~\ref{prop:h-GS}.

Consider an arbitrary $S \subset M$, and recall that $h(S) = \sum_{r=1}^{|S|} \mu(S_r,r)$. Consider the matching $MS$ that for every $1 \le r \le |S|$ matches item $S_r\in M$ with vertex $r \in V$. Claim~\ref{cl:exchange} implies that $MS$ is a maximum weight matching on the subgraph induced by $S$ and $V$. Consequently, $h'(S) = \sum_{r=1}^{|S|} \mu(S_r,r) = h(S)$, as desired.
\end{proof}
}

\noindent {\bf Remark:}
The proof of Proposition~\ref{prop:h-GS} does not use the specific expression of $\mu$ in Equation~\ref{eq:mu}, but only the fact that it satisfies the property specified in Claim~\ref{cl:exchange}. Thus, for any $\mu$ function satisfying this property, the function $h$ defined in Equation~\ref{eq:h} is monotone GS.

\vspace{0.1in}
We are now ready to prove Theorem~\ref{thm:ba-loglog}.

\begin{proof}
In this proof, it will be convenient for us to use natural logarithms, denoted by $\ln$. Observe that even though $\log m$ and $\ln m$ differ by a constant multiplicative factor, the values $\frac{\ln m}{\ln\ln m}$ and $\frac{\log m}{\log\log m}$ are the same, up to additive terms that become negligible as $m$ grows. Hence in the statement of Theorem~\ref{thm:ba-loglog} we may replace $\frac{\log m}{\log\log m}$ by $\frac{\ln m}{\ln\ln m}$. 

Consider an arbitrary BA function $f$.  By scaling, we may assume without loss of generality that the budget of $f$ is $1+ \frac{\ln\ln m}{\ln m}$. (Choosing this value, which is $1 + o(1)$ instead of the more natural value of~1, simplifies the definition of $T$.) 
	Consider the threshold function $T: [m] \rightarrow R$, where $T(1) = \frac{\ln\ln m}{\ln m}$, and for every $i > 1$ we have $T(i) = \frac{1}{i \ln  m}$. 
Observe that $\sum_{j\le m} T(j) \le 1 + \frac{\ln\ln m}{\ln m}$.

	Given the BA function $f$, let $g:M \rightarrow \reals^{\geq 0}$ be the induced function on the singletons, and let $h$ be the function defined in Eq.~(\ref{eq:h}) with respect to $g$ and $T$. $h$ is GS by Proposition~\ref{prop:h-GS}. 
	
	Consider an arbitrary set $S$ of items. {We slightly abuse notation and write $g$ and $T$ also for the corresponding extensions of $g$ and $T$ into additive set functions.} By the definitions of $h$ and $g$ we have $h(S) \le g(S)$. We also have that $g(S) = f(S)$, unless $g(S) > 1 + \frac{\ln\ln m}{\ln m}$, in which case $f(S) = 1 + \frac{\ln\ln m}{\ln m}$. However, $h(S) \le 1 + \frac{\ln\ln m}{\ln m}$ for every set $S$, because $T(S) \le 1 + \frac{\ln\ln m}{\ln m}$.	This shows that $h(S) \le f(S)$. It remains to show that $h(S) \ge f(S)\cdot (1 - o(1))\frac{\ln\ln m}{\ln m}$. 
	
	Recall that $S_1$ denotes the highest ranked item in $S$. We may assume that $g(S_1) < \frac{\ln\ln m}{\ln m}$, as otherwise we have $h(S) \ge h(S_1) \ge \frac{\ln\ln m}{\ln m} \ge (1 - o(1))\frac{\ln\ln m}{\ln m}f(S)$, as desired.
	
	Let $i$ be the smallest position for which $T(i) < \frac{\ln\ln m}{\ln m} g(S_i)$. (We may assume that such a position exists, as otherwise $h(S) \ge \frac{\ln\ln m}{\ln m} g(S) \ge \frac{\ln\ln m}{\ln m} f(S)$.) As $g(S_i) \le g(S_1) < \frac{\ln\ln m}{\ln m}$, the inequality $\frac{1}{i \ln  m} = T(i) < \frac{\ln\ln m}{\ln m} g(S_i)$ implies that $i > \frac{\ln m}{2(\ln\ln m)^2}$.  Let $k = \frac{2(\ln\ln m)^2}{\ln m}i$. Observe that $k > 1$, that $T(k) \le g(S_i) \le g(S_k)$, and that $\ln \frac{i}{k} = (1 - o(1))\ln\ln m$. It follows that
	
	$$h(S) \ge \sum_{k \le  j \le i} T(j) = \frac{1}{\ln m} \sum_{k \le  j \le i} \frac{1}{j} \simeq \frac{1}{\ln m} \ln \frac{i}{k} \simeq \frac{\ln\ln m}{\ln m} \ge (1 - o(1))\frac{\ln\ln m}{\ln m} f(S).$$ 
\end{proof}

\subsection{A Unified Approach for Approximating Concave Functions of (some) GS Functions}
\label{sec:concave}

To prove that a given function $f$ can be approximated by a GS function within a polylogarithmic ratio, we follow an approach that is fairly standard for approximating a function by a function from a different class. However, implementing this approach in the context of GS functions involves various subtleties that are not commonly encountered when other classes are concerned.

Let $f$ be a set function that is not GS, and we wish to find a GS function $h$ that approximates $f$ well. Our approach starts off as follows. For a parameter $T$ that is at most polylogarithmic in $m$, we find GS functions $h_1, \ldots, h_T$ with the following {\em sandwich property} holding for every set $S \subset M$:

$$\max_{t \le T} h_t(S)  \le f(S) \le \sum_{t \le T} h_t(S)$$

We remark that sometimes we use an extension of the sandwich property in which either the right hand side or the left hand side is multiplied by some constant, but this does not significantly affect the discussion below.

The sandwich property naturally suggests considering a lower bound function $h_L$ defined as $h_L(S) = \max_{t \le T} h_t(S)$ and an upper bound function $h_U$ defined as $h_U(S) = \sum_{t \le T} h_t(S)$. We have that $h_L(S) \le f(S) \le h_U(S)$, and for every $S$ we have $h_U(S) \le T \cdot h_L(S)$. Consequently, both $h_L$ and $\frac{1}{T}h_U$ approximate $f$ from below within a factor of $T$. The problem is that neither $h_L$ nor $h_U$ need to be GS, as GS is not preserved neither under the $\max$ operation (in fact, using the $\max$ operation one gets from GS the whole class XOS), nor under summation.

To overcome this problem, we use one of two approaches. The first approach is to select the functions $h_t$ in such a way that each function is supported over a subset of the items (items not in the support have value~0 under $h_t$), with no intersections between the subsets that support two different functions. Under this condition, the function $h_U$ is indeed GS, by the following observation (that can easily be verified using the triplet condition).

\begin{observation}
	\label{obs:gs-disjoint-sum}
	Let $h_1, \ldots, h_T$ be GS functions over $T$ pairwise disjoint sets of items $M_1, \ldots, M_T$. Then, the function $h_U=\sum_{t=1}^{T}h_t$ is GS over the set of items $\biguplus_{i=1}^{T}M_i$.
\end{observation}

In many interesting cases, we do not know how to implement the first approach (selecting the functions $h_t$ in such a way that each function is supported over a disjoint subset of the items) while keeping $T$ polylogarithmic in $m$. Consequently, we need to deal with the situation in which $h_U$ is not GS (and neither is $h_L$). Our second approach is to find a GS function that is sandwiched between $h_L$ and $h_U$. Luckily, the existence of such a function is guaranteed by the following argument. 

View each function $h_t$ as a valuation function of a bidder in a combinatorial auction, and consider the maximum welfare $W(S)$ function  for the auction with items $M$ and the above bidders. For every set $S$ of items, $W(S)$ is the maximum welfare that can be obtained by allocating $S$ to the bidders. (This function $W$ is also referred to as the {\em convolution} of the collection of functions $h_t$, {where Definition~\ref{def:conv} is the special case when the collection has only two functions}.) As each $h_t$ is GS, so is the welfare function $W$ (see~\cite{lehmann2006combinatorial,PaesLeme17,Murota96}). Moreover, $\max_{t \le T} h_t(S) \le W(S) \le \sum_{t \le T} h_t(S)$. As both $f$ and $H$ enjoy the same sandwich property, it follows that $\frac{1}{T}W$ approximates $f$ from below by a ratio no worse than $T^2$. (Note that here the approximation ratio is $T^2$, whereas in the first approach it was $T$. The difference stems from the following fact. In both approaches we identify a GS function that is within a factor of $T$ of $f$, but in the first approach this function $h_U$ satisfies $\frac{1}{T}h_U(S) \le f(S) \le h_U(S)$ for every $S$, whereas in the second approach this function $W$ satisfies $\frac{1}{T}W(S) \le f(S) \le T \cdot W(S)$ for every $S$.) 

The above discussion is summarised (and slightly generalized to an extended sandwich property) in the following lemma (whose proof can easily be completed by the reader, given the above discussion).

\begin{lemma}
	\label{lem:approach}
	Suppose that for a function $f$, there are 
	{$0< \alpha \le 1$ and $\beta \ge 1$} and a collection of GS functions $h_1, \ldots, h_T$ satisfying the following extended sandwich property for every set $S \subset M$:
	
	$$\alpha \cdot \max_{t \le T} h_t(S)  \le f(S) \le \beta \cdot \sum_{t \le T} h_t(S)$$
	
	Let $h_U$ be the function defined as $h_U(S) = \sum_{t \le T} h_t(S)$, and let $W$ be the welfare function (the convolution of the functions $h_t$) as defined above. Then:
	
	\begin{enumerate}
		\item $W$ is a GS function,  and the GS function $\frac{\alpha}{T}W$ approximates $f$ from below within a ratio of $\frac{\beta}{\alpha}T^2$.
		\item If every item is in the support of at most one of the functions $h_t$, then $h_U$ is a GS function, and the GS function $\frac{\alpha}{T}h_U$ approximates $f$ from below within a ratio of $\frac{\beta}{\alpha}T$.
	\end{enumerate}
\end{lemma}

As an illustration of the use of Lemma~\ref{lem:approach}, we prove Proposition~\ref{prop:ba-log} (which proves a weaker bound than that proved in Theorem~\ref{thm:ba-loglog}, but does so via a simpler proof).

\begin{proposition}
	\label{prop:ba-log}
	Every BA function $f$ can be approximated by a GS function within a ratio of $O(\log m)$.
\end{proposition}

\begin{proof}
	Suppose without loss of generality that the budget limitation of $f$ is~1, that there is no item of value larger than~1, and that $m$ is a power of~2. Round the value of each item down to the nearest power of~2 (namely, to~1, to $\frac{1}{2}$, and so on), thus obtaining $f'$. For every set $S$ we have that $f'(S) \le f(S) \le 2f'(S)$. Partition items into $T = 1 + \log m$ classes $M_0, \ldots M_{\log m}$ by their value, where all items of value $2^{-t}$ are in class $M_t$, and class $M_{\log m}$ contains also all items of value smaller than $\frac{1}{m}$. Let $h_t$ be the function $f'$ restricted to the items of $M_t$. Observe that $h_{\log m}$ is an additive function (as the sum of item values cannot reach the budget), and hence GS.  Each of $h_t$ for $t < \log m$ remains BA, but is also GS. This is because all items in $M_t$ for $t < \log m$ have the same value, and hence the triplet condition holds even with the budget limit. We have the extended sandwich property:
	
	$$\max_{0 \le t \le \log m} h_t(S)  \le f(S) \le 2\sum_{0 \le t \le \log m} h_t(S)$$
	
	Let $h_U$ be the function satisfying $h_U(S) = \sum_{0 \le t \le \log m} h_t(S)$. By item~2 of Lemma~\ref{lem:approach}, the function $\frac{1}{T}h_U(S)$ is GS and approximates $f$ (from below) within a ratio of $2(1 + \log m)$. 
\end{proof}


Lemma~\ref{lem:approach} does not suffice for the proof of Theorem~\ref{thmWMM} (approximating any concave function $g(f)$ of a Rado function $f$ by GS within a ratio of $O(\log^2 m)$). One problem is that the range of values of $f$ can span values that differ from each other by more than a polynomial factor in $m$. (The function $g(f)$ that we wish to approximate may still have a range of values that is polynomial in $m$, if composing the concave function $g$ with $f$ shrinks the range of values.) Handling such situations requires a version of Lemma~\ref{lem:approach} in which $T$ is polynomial in $m$ rather than logarithmic in $m$. Lemma~\ref{lem:approach1}, whose proof is considerably more complicated than the proof of Lemma~\ref{lem:approach}, provides such a version. 

\begin{lemma}
	\label{lem:approach1}
	Let $H = \{h_0, \ldots, h_T\}$ be a collection of submodular functions, where for every function $h_t$, all marginals are either~0 or~$2^t$. (This condition is known to imply that $h_t$ is an MRF, scaled by $2^t$.)
	Suppose that for a valuation function $f$, there are there are 
{$0< \alpha \le 1$ and $\beta \ge 1$} satisfying the following extended sandwich property for every set $S \subset M$:
	
	$$\alpha \cdot \max_{t \le T} h_t(S)  \le f(S) \le \beta \cdot \sum_{t \le T} h_t(S)$$
	
	Let $g$ be an arbitrary normalized monotone concave function. Then:
	
	\begin{enumerate}
		\item The function $g(f)$ can be approximated by a GS function within a ratio of $O(\frac{\beta}{\alpha} (\log m)^2)$.
		\item If every item is in the support of at most one of the functions $h_t$, then the function $g(f)$ can be approximated by a GS function within a ratio of $O(\frac{\beta}{\alpha} \log m)$.
	\end{enumerate}
\end{lemma}

The proof of Lemma~\ref{lem:approach1} is deferred to Appendix~\ref{sec:advanced-lemma}.

\subsection{Approximability of Concave of Rado within a Ratio $O(\log^{2} m)$}
\label{sec:app-positive}



In this section we prove Theorem~\ref{thmWMM}, concerning approximating a concave function of a Rado function by a GS function within a ratio of $O(\log^2 m)$. 

Before doing so, we consider several natural subclasses of Rado functions, namely matroid rank functions, additive functions, and weighted matroid rank functions.
For additive and weighted matroid rank functions we provide a better approximation ratio of $O(\log m)$.
For matroid rank functions, it is a known fact that applying a concave function to such functions results in a GS function; we provide a proof for completeness. 



\begin{proposition}
	\label{lem:concave-mrf-gs}
	Let $f$ be a matroid rank function. Let $g$ be a function obtained by composing a concave function with $f$. Then $g$ is GS.
\end{proposition}

\begin{proof}
	Clearly, $g$ is submodular; thus, by the local characterization of GS functions it suffices to prove Inequality (\ref{eq:no-unique-max}).
	If $f(j,k|S)=0$, then (\ref{eq:no-unique-max}) follows by the DC property of matroids.
	Similarly, if $f(j,k|S)=1$, then all terms in the RHS of (\ref{eq:no-unique-max}) are at least $g(1)$, and the inequality follows.
	It remains to prove the inequality for the case where $f(j,k|S)=2$.
	By the DC property of matroids, $f(j|S)=f(k|S)=1$.
	If $f(i|S)=1$, then by the exchange property of matroids, either $f(i,j|S)=2$ or $f(i,k|S)=2$. The LHS of (\ref{eq:no-unique-max}) is $g(1)+g(2)$, and one of the terms in the RHS of (\ref{eq:no-unique-max}) is at least $g(1)+g(2)$ as well, as desired.
	If $f(i|S)=0$, then the LHS of (\ref{eq:no-unique-max}) is $g(0)+g(2)$, and both terms in the RHS of (\ref{eq:no-unique-max}) are $g(1)+g(1)$. By concavity of $g$, $g(0)+g(2) \leq g(1)+g(1)$, as desired.
\end{proof}

Proposition~\ref{lem:concave-mrf-gs} turns out to be very useful in our context, because every submodular function (and hence also every GS function) in which the marginals are either 0 or 1 is an MRF. It will be used in the proof of Lemma~\ref{lem:approach1}, where we encounter GS functions that are scaled versions of MRFs (marginals are~0 or $c$ for some value $c$). Consequently, applying a concave function on these functions also results in a GS function.

We next consider additive and weighted matroid rank functions.

\begin{proposition}
\label{pro:additive}
	Let $g$ be a normalized monotone concave function and let $f$ be an additive set function. The function $g(f)$ can be approximated by a GS function within a ratio of $O(\log m)$.
\end{proposition}

\begin{proof}
	Let $R$ be the ratio between the maximum value and minimum value item in $f$. Without loss of generality, we assume that the smallest value that an item has is~1.

Round the value of each item down to the nearest power of~2, thus obtaining $f'$. For every set $S$ we have that $f'(S) \le f(S) \le 2f'(S)$. By monotonicity and concavity of $g$, it also holds that $g(f'(S)) \le g(f(S)) \le 2g(f'(S))$. Consequently, for simplicity of notation and losing only a factor of~2 in the approximation ratio, we assume that in $f$ the value of every item is a power~2.

Partition items into $T = 1 + \log R$ classes $M_0, \ldots M_{\log R}$ by their value, where all items of value $2^{t}$ are in class $M_t$. For every $0 \le t \le \log R$, let $f_t$ be the the function defined by $f_t(S) = f(S \cap M_t)$. 
Then for every set $S\subset M$ we have the sandwich property
$$\max_{0 \le t \le R} f_t(S) \le f(S) \le \sum_{0 \le t \le R} f_t(S)$$ 
The proof of the proposition now follows from item~2 of Lemma~\ref{lem:approach1}.
\end{proof}

\begin{proposition}
	\label{pro:wmrf}
	Let $g$ be a normalized monotone concave function and let $f$ be a weighted MRF function. $g(f)$ can be approximated by a GS function within a ratio of $O(\log m)$.
\end{proposition}

\begin{proof}
The proof follows that of Proposition~\ref{pro:additive}, but with one change. Recall the classes $M_t$ defined in the proof of Proposition~\ref{pro:additive}. Every item in the class $M_t$ had $f$ value $2^t$. In the proof of Proposition~\ref{pro:additive}, $f$ restricted to the items of $M_t$ was an additive function. In our current context, when $f$ is a WMRF, $f$ restricted to the items of $M_t$ is an MRF (scaled by $2^t$). Consequently, item~2 of Lemma~\ref{lem:approach1} still applies.
\end{proof}

We now prove Theorem~\ref{thmWMM}, establishing that if $g$ is a normalized monotone concave function and $f$ is a Rado function, then the function $g(f)$ can be approximated by a GS function within a ratio of $O(\log^{2} m)$.

Recall the definition of Rado valuations. 
A set function $f:2^M \rightarrow \reals^+$ is a Rado valuation if there is a bipartite graph $G(M,V;E)$ with non-negative weights on the edges, and a matroid $\mathcal{M}=(V,\mathcal{I})$, such that for every $S\subseteq M$, $f(S)$ is the total weight of the max weight matching on the subgraph induced by $S$ and a subset of $V$ that belongs to $\mathcal{I}$.

\begin{proof} (of Theorem~\ref{thmWMM})
Let $f$ be a Rado function. Recall that its representation involves a bipartite graph $G(U,V;E)$ with non-negative weights on the edges, with $m = |U|$ left side vertices, that we shall refer to as {\em items}.  

Let $R$ be the ratio between the maximum weight and minimum weight edge in $G$. Without loss of generality, we assume that the smallest weight that an edge has is~1, and that $R$ is a power of~2. 

Round down the weight of each edge to the nearest power of~2, thus obtaining $f'$. For every set $S$ we have that $f'(S) \le f(S) \le 2f'(S)$. By monotonicity and concavity of $g$, it also holds that $g(f'(S)) \le g(f(S)) \le 2g(f'(S))$. Consequently, for simplicity of notation and losing only a factor of~2 in the approximation ratio, we assume that in $G$ the weight of every edge is a power~2.


Partition the graph into a collection $\{G_t\}$ of uniform weight graphs, where for every $t$ graph $G_t$ contains all edges of weight $2^t$. 
The corresponding Rado function for each such graph $G_t$ will be referred to as $f_t$, and it is a matroid rank function (scaled by $2^t$, because $f_t$ is submodular, and its marginals are~0 and~$2^t$). 

We have the sandwich property $\max_{t\le \log R} f_t(S) \le f(S) \le \sum_{t\le \log R} f_t(S)$. 
Using item~1 of Lemma~\ref{lem:approach1} we have that there is a GS function that approximates $g(f)$ within a ratio of $O((\log m)^2)$. 
\end{proof}

\section*{Acknowledgements}

{In a previous version of this paper, Lemma~\ref{lem:twoItemSymmetrization} was proved using a computer assisted proof. We are grateful to Noam Guterman and Ittay Toledo who each independently wrote code that was used in that previous proof.}

\bibliographystyle{plain}
\bibliography{bib-file}

\medskip

\appendix
\section*{APPENDIX}
\section{Appendix for Section~\ref{sec:preliminaries}}
\label{sec:app-prelim}

\begin{proof}[Proof of Lemma~\ref{lem:local-approx}:]
	We distinguish between two cases.
	
	\noindent Case (i): $f(i,j|S) < f(k|S)$. First observe that
	$\max\{f(i|S), f(j|S)\} \ge \frac{1}{2} (f(i|S) + f(j|S)) \geq \frac{1}{2}f(i,j|S)$,
	where the last inequality follows by submodularity.
	Combining this with monotonicity and the assumption of case (i), we get that
	$\max\{f(i,k|S) + f(j|S), f(j,k|S) + f(i|S)\} \geq \max\{f(k|S) + f(j|S), f(k|S)+ f(i|S)\} = f(k|S) + \max\{f(j|S), f(i|S)\} \geq f(k|S)+\frac{1}{2}f(i,j|S) \geq \frac{3}{4}(f(i,j|S) + f(k|S))$.
	
	\noindent Case (ii): $f(i,j|S) \geq f(k|S)$.
	By submodularity, $f(i,k|S) + f(j,k|S) \ge f(i,j,k|S) + f(k|S)$ and also $f(i|S) + f(j|S) \ge f(i,j|S)$. By these inequalities and monotonicity we get that $f(i,k|S) + f(j|S) + f(j,k|S) + f(i|S) \ge 2f(i,j|S) + f(k|S)$. Hence $\max[f(i,k|S) + f(j|S), f(j,k|S) + f(i|S)] \ge f(i,j|S) + \frac{1}{2}f(k|S) \ge \frac{3}{4}(f(i,j|S) + f(k|S))$, where the last inequality follows by the assumption of case (ii).
	
	To show that $\frac{3}{4}$ is the best possible ratio, consider a budget additive function $f$ with $f(i) = f(j) = 1$, $f(k)=2$, and a budget of~2.
	For $S=\emptyset$, it holds that $\max\{f(i,k|S) + f(j|S), f(j,k|S) + f(i|S)\}=3$, whereas $f(i,j|S) + f(k|S))=4$.
\end{proof}

\subsection{Matroids and Additional Valuation Functions}
\label{sec:matroids-prelim}

\newcommand{\mat}{\mathcal{M}}
\newcommand{\ind}{\mathcal{I}}

A {\em matroid} $\mat$ is a pair $(M,\ind)$, where $M$ is a finite set of elements, and $\ind$ is a non-empty collection of subsets of $M$ (often termed the collection of {\em independent sets})  satisfying the following two conditions: 
\begin{itemize}
	\item Downward-closed (DC): If $S \in \ind$ and $S' \subseteq S$, then $S' \in \ind$.
	\item Exchange property: For any two sets $S,T \in \ind$ such that $|S|<|T|$, there exists an element $j \in T \setminus S$ such that $S \cup \{j\} \in \ind$.
\end{itemize}

Given a matroid $\mat=(M,\ind)$, the rank function of $\mat$ is a function $rank_{\mat}:2^{M}\rightarrow \mathbb{N}$, where $rank_{\mat}(S) = \max_{T\in \ind}|S \cap T|$.  


\begin{definition}
	\label{def:mrf}
A set function $f:2^{M}\rightarrow \reals$ is a {\em matroid rank function} (MRF) if there exists a matroid $\mat = (M,\ind)$ such that $f(S) = rank_{\mat}(S)$ for every $S\subset M$.
\end{definition}

A more general class of valuations is {\em weighted matroid rank} (WMRF) functions. 
Given a matroid $\mat=(M,\ind)$ and a weight function $w$ that associates a non-negative weight $w_j$ with every element $j \in M$, the weighted rank function of $\mat$ with respect to $w$ is a function $wrank_{\mat,w}:2^{M} \rightarrow \reals$, where 
$wrank_{\mat}(S) = \max_{T\in \ind, T \subseteq S}\sum_{j \in T}w_j$.
%

\begin{claim}\cite{Oxley2011}
	\label{cl:mrf-sm-01}
	A valuation function is MRF if and only if it is a submodular valuation with binary marginal values. 
\end{claim}

With a slight abuse of notation, we shall refer to a submodular valuation with marginal values in $\{0,c\}$ (for some constant $c$) as MRF as well.

\begin{definition}
	\label{def:wmrf}
	A set function $f:2^{M}\rightarrow \reals$ is a {\em weighted matroid rank function} (WMRF) if there exist a matroid $\mat = (M,\ind)$, and a weight function $w$ over the elements in $M$ such that $f(S)=wrank_{\mat,w}(S)$ for every $S \subseteq M$.
\end{definition}

Finally we define the class of {\em matroid rank sum} (MRS), which is the class of sum over matroid rank functions.

\begin{definition}
	\label{def:mrs}
	A set function $f:2^{M}\rightarrow \reals$ is a {\em matroid rank sum} (MRS) if there exists a collection of MRF functions $f_1, \ldots f_k : 2^{M}\rightarrow \reals$, and associated non-negative weights $w_1, \ldots, w_k \in \reals^+$ such that $f(S) = \sum_{j=1}^{k}w_{j}f_j(S)$ for every $S \subset M$.
\end{definition}

Note that coverage valuations are MRS over matroid rank functions of rank 1. 

\begin{definition}
	\label{def:oxs}
	A set function $f:2^{M}\rightarrow \reals$ is {\em OXS} if there exists a bipartite graph $G(M,V;E)$ with non-negative weights on the edges, such that for every $S \subseteq M$, $f(S)$ is the total weight of the maximum weighted matching on the subgraph induced by $S$ and $V$.
\end{definition}

{OXS functions form a subclass of GS functions~\cite{lehmann2006combinatorial}.} 

\section{Gap of $\frac{2k}{k+1}$ between GS and Submodular via the Greedy Characterization}
\label{sec:factor2-sep}
In this section we show that the submodular (and in fact, budget additive, BA) function $BA(k,1)$ (see Definition~\ref{def:ba-k-d}) cannot be approximated by a GS function within a better factor than $\frac{2k}{k+1}$. We do so without making use of symmetrization (Theorem~\ref{thm:symmetrization}).
%
Hereafter, $\rho_f$ denotes the best ratio by which any GS function can approximate $f$.

\begin{proposition}
	\label{prop:approx-2}
	Let $f$ be the budget additive function with budget~1, one item $a$ of value~1, and a set $B$ of $k$ items $b_1, \ldots, b_k$, each of value $1/k$.
	Then $\rho_f = \frac{2k}{k+1}$.
\end{proposition}

\begin{proof}
	For presentation simplicity, for an item $j$, we write $g(j)$ instead of $g(\{j\})$ to denote its value under valuation $g$.
	
	To see that $\rho_f \le \frac{2k}{k+1}$, consider the function $g$, where $g(a)=\frac{k+1}{2k}$, $g(b_j)=\frac{1}{k}$ for every $j \in [k]$, and the marginal value of every additional item from $B$ is $\frac{1}{2k}$ (except for the last item that contributes 0). 
	That is, for a set $S \subseteq B$ of size $\ell \geq 1$, $g(S)=\frac{\ell+1}{2k}$; for a set $S$ that contains $a$ and $\ell$ items from $B$, $g(S)=\min[\frac{k+1+\ell}{2k},1]$.
	Clearly, $g$ approximates $f$ within a ratio $\frac{2k}{k+1}$. 
	
	The function $g$ is clearly submodular. To see that it is GS, we show that it satisfies the triplet condition.
	Since the items in $B$ are symmetric, it suffices to prove that for every $i,j \in \{1,\ldots,k\}$ and item set $S$ s.t. $a,b_i,b_j \not\in S$, $g(a\mid S)+g(b_ib_j\mid S)=g(b_i\mid S)+g(ab_j\mid S)$.
	If $S=\emptyset$, then $g(a\mid S)+g(b_ib_j\mid S)=g(b_i\mid S)+g(ab_j\mid S)=\frac{k+4}{2k}$. 
	Else ($1 \leq |S| \leq k-2)$, $g(a\mid S)+g(b_ib_j\mid S)=g(b_i\mid S)+g(ab_j\mid S)=\frac{k+2}{2k}$.
	It follows that $g$ is GS.
	
	To prove that $\rho_f \ge \frac{2k}{k+1}$, let $g$ be a GS function that approximates $f$ from below. 
	W.l.o.g., $g(b_1) \ge g(b_j)$ for all $2 \le j \le k$. Also, we may assume that $g(a) > g(b_1)$, as otherwise the approximation ratio of $g$ is no better than $k \ge 2 > \frac{2k}{k+1}$.
	
	Consider a vector of prices $\prices$ with $p_{a} = g(a) - g(b_1)$, and $p_{b_j} = 0$ for all $j \in [k]$. The greedy algorithm first selects $a$ (and pays $g(a) - g(b_1)$), and then the remaining items. As $g$ is upper bounded by~1, the profit is at most $1 - g(a) + g(b_1)$. Alternatively, one can select the set $B$ and obtain a profit of $g(B)$. By Lemma~\ref{lem:greedy-char} (the greedy characterization of GS functions), we have that $1 - g(a) + g(b_1) \ge g(B)$, thus $1 + g(b_1) \ge g(B) + g(a)$. Since $g(b_1)\leq f(b_1) = \frac{1}{k}$ it follows that $\frac{k+1}{k} \ge g(B) + g(a)$. Consequently, $\min[g(a),g(B)] \le \frac{k+1}{2k}$, whereas $f(a) = f(B) = 1$.
\end{proof}


\section{Induction By Networks}
\label{appendix:induction-by-networks}

In Section \ref{sec:log-gap} we identified two operations (max-symmetrization and partial-symmetrization) that preserve GS.
Those form a special case of an operations called \emph{induction by networks} introduced by Kobayashi,  Murota and Tanaka \cite{kobayashi2007operations} in the context of $M$-convex functions on jump systems. Although the operation they identified is not explicitly about GS, their results can be mapped to an operation on GS. Here we provide a version of their definition as well as a direct proof (not going through $M$-convexity and jump systems) that it preserves GS. The proof will follow as a consequence of Claim \ref{claim:partial}.

\begin{definition}\label{def:induction_networks}
Consider a bipartite graph $G(U,V;E)$ with weights $w_e \in \reals$ for each edge $e \in E$ and a set function $v : 2^V \rightarrow \reals$. Given a  subset $M \subset E$ define $\partial_U(M) \subseteq U$ and $\partial_V(M) \subseteq V$ as the set of incident vertices on $U$ and $V$ respectively. With that we can define the induction of $v$ by $G$ as the function $f:2^U \rightarrow \reals$ such that for each $S \subseteq U$
$$f(S) = \max_{\text{matching } M \text{ s.t. } \partial_U(M) \subseteq S} \left[ v(\partial_V (M)) + \sum_{e \in M} w_e \right]$$
\end{definition}

We note that both max-symmetrization and partial-symmetrization are special cases of induction by the graphs in Figure \ref{fig:two_examples_of_induction}.

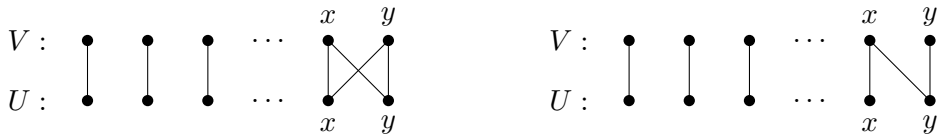
\begin{figure}[h]
\centering
\begin{tikzpicture}[scale=.8, inner sep=1.5pt] 
 \begin{scope}
 \node[circle,fill] at (0,0) {};
 \node[circle,fill] at (1,0) {};
 \node[circle,fill] at (2,0) {};
 \node at (3,0) {$\hdots$};
 \node[circle,fill] at (4,0) {};
 \node[circle,fill] at (5,0) {};
 \node at (4,-.4) {$x$};
 \node at (5,-.4) {$y$};
 
 \node[circle,fill] at (0,1) {};
 \node[circle,fill] at (1,1) {};
 \node[circle,fill] at (2,1) {};
 \node at (3,1) {$\hdots$};
 \node[circle,fill] at (4,1) {};
 \node[circle,fill] at (5,1) {};
 \node at (4,1.4) {$x$};
 \node at (5,1.4) {$y$};
 \node at (-1,0) {$U:$};
 \node at (-1,1) {$V:$};
 
 \draw (0,0)--(0,1);
 \draw (1,0)--(1,1);
 \draw (2,0)--(2,1);
 \draw (4,0)--(4,1);
 \draw (4,0)--(5,1);
 \draw (5,0)--(5,1);
 \draw (5,0)--(4,1);
\end{scope}

 \begin{scope}[xshift=9cm]
 \node[circle,fill] at (0,0) {};
 \node[circle,fill] at (1,0) {};
 \node[circle,fill] at (2,0) {};
 \node at (3,0) {$\hdots$};
 \node[circle,fill] at (4,0) {};
 \node[circle,fill] at (5,0) {};
 \node at (4,-.4) {$x$};
 \node at (5,-.4) {$y$};
 
 \node[circle,fill] at (0,1) {};
 \node[circle,fill] at (1,1) {};
 \node[circle,fill] at (2,1) {};
 \node at (3,1) {$\hdots$};
 \node[circle,fill] at (4,1) {};
 \node[circle,fill] at (5,1) {};
 \node at (4,1.4) {$x$};
 \node at (5,1.4) {$y$};
 \node at (-1,0) {$U:$};
 \node at (-1,1) {$V:$};

 \draw (0,0)--(0,1);
 \draw (1,0)--(1,1);
 \draw (2,0)--(2,1);
 \draw (4,0)--(4,1);
 \draw (5,0)--(5,1);
 \draw (5,0)--(4,1);
\end{scope}

\end{tikzpicture}
\caption{Max-symmetrization (left) and partial symmetrization (right) as inductions by the {corresponding} graphs with zero weight on the edges.}
\label{fig:two_examples_of_induction}
\end{figure}

\begin{theorem}\label{thm:induction_networks}
In the context of the Definition \ref{def:induction_networks}, if function $v$ is GS, then its induction by $G$ is also GS.
\end{theorem}

We will prove Theorem \ref{thm:induction_networks} by writing induction by networks as a composition of three elementary operations that preserve GS. We will borrow the terminology in \cite{kobayashi2007operations} and call the first operation \emph{splitting}. Recall that we are using the notation $Sx$ to represent $S \cup \{x\}$.\\

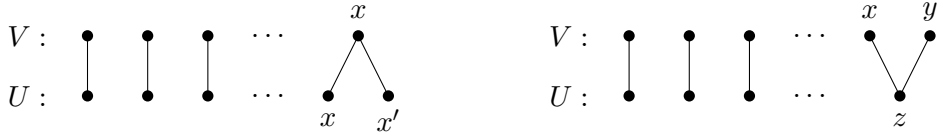
\begin{figure}[h]
\centering
\begin{tikzpicture}[scale=.8, inner sep=1.5pt] 
 \begin{scope}[xshift=9cm]
 \node[circle,fill] at (0,0) {};
 \node[circle,fill] at (1,0) {};
 \node[circle,fill] at (2,0) {};
 \node at (3,0) {$\hdots$};
 \node[circle,fill] at (4.5,0) {};
 \node at (4.5,-.4) {$z$};

 \node[circle,fill] at (0,1) {};
 \node[circle,fill] at (1,1) {};
 \node[circle,fill] at (2,1) {};
 \node at (3,1) {$\hdots$};
 \node[circle,fill] at (4,1) {};
 \node[circle,fill] at (5,1) {};
 \node at (4,1.4) {$x$};
 \node at (5,1.4) {$y$};
 \node at (-1,0) {$U:$};
 \node at (-1,1) {$V:$};
 
 \draw (0,0)--(0,1);
 \draw (1,0)--(1,1);
 \draw (2,0)--(2,1);
 \draw (4.5,0)--(4,1);
 \draw (4.5,0)--(5,1);
 \end{scope}

 \begin{scope}
 \node[circle,fill] at (0,0) {};
 \node[circle,fill] at (1,0) {};
 \node[circle,fill] at (2,0) {};
 \node at (3,0) {$\hdots$};
 \node[circle,fill] at (4,0) {};
 \node[circle,fill] at (5,0) {};
 \node at (4,-.4) {$x$};
 \node at (5,-.4) {$x'$};
 
 \node[circle,fill] at (0,1) {};
 \node[circle,fill] at (1,1) {};
 \node[circle,fill] at (2,1) {};
 \node at (3,1) {$\hdots$};
 \node[circle,fill] at (4.5,1) {};
 \node at (4.5,1.4) {$x$};
 \node at (-1,0) {$U:$};
 \node at (-1,1) {$V:$};

 \draw (0,0)--(0,1);
 \draw (1,0)--(1,1);
 \draw (2,0)--(2,1);
 \draw (4,0)--(4.5,1);
 \draw (5,0)--(4.5,1);
\end{scope}

\end{tikzpicture}
\caption{Splitting (left) and aggregation (right).}
\label{fig:splitting_and_aggregation}
\end{figure}

\begin{lemma}[Splitting]
Given a GS function $v:2^V \rightarrow \reals$ and an item $x \in V$, create a copy $x'$ of item $x$ and define $f:2^{Vx'} \rightarrow \reals$ such that for all sets $S \subseteq V$ we have $f(S) = v(S)$ and $f(Sx') = v(Sx)$. Then function $f$ is also GS.
\end{lemma}

\begin{proof}
Since items $x$ and $x'$ are identical and perfect substitutes we can obtain the demand of $f$ under any given price from the demand of $v$. 
To compute the demand of $f$ under a price vector $\mathbf{p} \in \mathbb{R}^{Vx'}$, first let $\mathbf{\tilde p}$ be a vector in $\mathbb{R}^V$ where $\tilde p_{x} = \min (p_x, p_{x'})$ and $\tilde p_{i} = p_i$ for any other $i$. Now, $D_f(\mathbf{p})$ can be obtained by $D_v(\mathbf{\tilde p})$ by replacing $x$ in each set $S \in D_v(\mathbf{\tilde p})$ by $x$ or $x'$, whichever is cheaper. If they have the same price, replace each set containing $x$ by two sets, one containing $x$ and the other containing $x'$. 
Hence, the characterization of GS in Definition \ref{def:gs_kc} directly extends from $v$ to $f$.
\end{proof}

\begin{lemma}[Aggregation]
Given a GS function $v:2^V \rightarrow \reals$ and items $x,y \in V$, consider an extra item $z$ and define $U = Vz \setminus \{x,y\}$. Now define  $f:2^{U} \rightarrow \reals$ such that for all sets $S \subseteq U \setminus \{z\}$ we have $f(S) = v(S)$ and $f(Sz) = \max[v(Sx), v(Sy)]$. Then function $f$ is also GS.
\end{lemma}

\begin{proof}
Let $q$ be the partial $(x,y)$-symmetrization of $v$ defined in Claim \ref{claim:partial}. The function $q$ is known to be GS by that claim. Now restrict $q$ to $V \setminus \{y\}$ and rename item $x$ to $z$. The obtained function is exactly $f$. Since restriction and item-renaming preserve GS, then $f$ is also GS.
\end{proof}

\begin{lemma}[Additive perturbation]
Given a GS function $v:2^V \rightarrow \reals$ and weights $w_i \in \reals$ for each $i \in V$, define $f:2^V \rightarrow \reals$ such that $f(S) = v(S) + \sum_{i \in S} w_i$. Then function $f$ is also in $GS$.
\end{lemma}

\begin{proof}
It follows directly from Lemma \ref{lem:triplet-char} since in each side of equation \eqref{eq:no-unique-max} a term $w_i + w_j + w_k$ will be added so they will cancel out. Hence for any weights $w_i$ (positive or negative) equation \eqref{eq:no-unique-max} holds for $f$ if and only if it holds for $v$. 
\end{proof}

\begin{proof}[Proof of Theorem \ref{thm:induction_networks}]
Induction by a bipartite graph $G$ is equivalent to an application of splitting, additive perturbation and aggregation. This can be done in three steps. See Figure \ref{fig:proof_of_induction_thm} for an illustration of the procedure. In the first step, apply splitting repeatedly to produce $k_v$ copies of each item $v \in V$ where $k_v$ is the number of edges  incident to $v$ in $G$. This leads to a function $f_1:2^E \rightarrow \reals$ given by $f_1(T) = v(\partial_V(T))$ for all $T \subseteq E$.

In the second step, apply the additive perturbation technique to obtain $f_2:2^E \rightarrow \reals$ such that $f_2(T) = f_1(T) + \sum_{e \in T} w_e$.

Finally, apply aggregation repeatedly so that all edges $e$ that have the same endpoint in $U$ are aggregated to the same item. Rename this aggregated item to the corresponding element $u \in U$. Let $f:2^U \rightarrow \reals$ be the function obtained by this procedure. Note that $f$ is exactly the induction of $v$ by $G$. Since each operation preserves GS, then its composition is also GS.
\end{proof}


\begin{figure}[h]
\centering
\begin{tikzpicture}[scale=.8, inner sep=1.5pt] 
 \begin{scope}
 \node[circle,fill] at (0,0) {};
 \node[circle,fill] at (1,0) {};
 \node[circle,fill] at (2,0) {};

 \node[circle,fill] at (0,1) {};
 \node[circle,fill] at (1,1) {};
 \node[circle,fill] at (2,1) {};
 \node at (-1,0) {$U:$};
 \node at (-1,1) {$V:$};
 
 \draw (0,0)--(0,1);
 \draw (1,0)--(1,1);
 \draw (1,0)--(0,1);
 \draw (2,0)--(2,1);
 \draw (2,0)--(1,1);
 \draw (2,0)--(0,1);
 
 \node at (0,1.4) {$a$};
 \node at (1,1.4) {$b$};
 \node at (2,1.4) {$c$};
 \node at (0,-.4) {$x$};
 \node at (1,-.4) {$y$};
 \node at (2,-.4) {$z$};

 \end{scope}

 \begin{scope}[xshift=7cm,yshift=-1.5cm]
 \node at (1,4.4) {$a$};
 \node at (3.5,4.4) {$b$};
 \node at (5,4.4) {$c$};

 \node[circle,fill] at (1,4) {};
 \node[circle,fill] at (3.5,4) {};
 \node[circle,fill] at (5,4) {};

 \node[circle,fill] at (0,3) {};
 \node[circle,fill] at (1,3) {};
 \node[circle,fill] at (2,3) {};
 \node[circle,fill] at (3,3) {};
 \node[circle,fill] at (4,3) {};
 \node[circle,fill] at (5,3) {};

 \node at (-1,4) {$V:$};
 \node at (-1,3) {$E:$};
 \node at (-1,1) {$E:$};
 \node at (-1,0) {$U:$};
  
 \draw (1,4)--(0,3);
 \draw (1,4)--(1,3);
 \draw (1,4)--(2,3);
 \draw (3.5,4)--(3,3);
 \draw (3.5,4)--(4,3);
 \draw (5,4)--(5,3);

 \node[circle,fill] at (0,1) {};
 \node[circle,fill] at (1,1) {};
 \node[circle,fill] at (2,1) {};
 \node[circle,fill] at (3,1) {};
 \node[circle,fill] at (4,1) {};
 \node[circle,fill] at (5,1) {};

 \draw[dashed] (0,1)--(0,3);
 \draw[dashed] (1,1)--(1,3);
 \draw[dashed] (2,1)--(2,3);
 \draw[dashed] (3,1)--(3,3);
 \draw[dashed] (4,1)--(4,3);
 \draw[dashed] (5,1)--(5,3);
 
 \node at (0,2) {\small{$w_{ax}$}};
 \node at (1,2) {\small{$w_{ay}$}};
 \node at (2,2) {\small{$w_{az}$}};
 \node at (3,2) {\small{$w_{by}$}};
 \node at (4,2) {\small{$w_{bz}$}};
 \node at (5,2) {\small{$w_{cz}$}};

 \node[circle,fill] at (0,0) {};
 \node[circle,fill] at (2,0) {};
 \node[circle,fill] at (3.5,0) {};
 \node at (0,-.4) {$x$};
 \node at (2,-.4) {$y$};
 \node at (3.5,-.4) {$z$};

 \draw (0,0) -- (0,1);
 \draw (2,0) -- (1,1);
 \draw (2,0) -- (3,1);
 \draw (3.5,0) -- (2,1);
 \draw (3.5,0) -- (4,1);
 \draw (3.5,0) -- (5,1);
 
 

\end{scope}

\end{tikzpicture}
\caption{Graph $G$ (left) and decomposition in the proof of Theorem \ref{thm:induction_networks} (right).}
\label{fig:proof_of_induction_thm}
\end{figure}
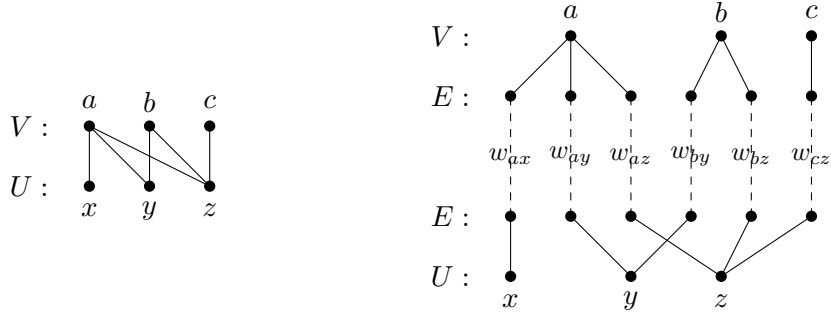

\section{(Non)-Closedness of SWS Valuations}
\label{app:sws-closed}
\begin{claim}
	The class of SWS functions is not closed under average, neither is it closed under convolution.
\end{claim}

\begin{proof}
	To see that SWS is not closed under average, consider two budget additive functions $f_1,f_2$ over 3 items $a,b,c$, where $f_1$ gives value 2 to $a,b$ and value 0 to $c$, and $f_2$ gives value 2 to $a,c$ and value 0 to $b$. Both functions have budget $2$. $f_1$ and $f_2$ are SWS, but the average of them (which gives value $2$ to $a$ and value $1$ to $b,c$) is not SWS, as $b$ and $c$ belong to the same symmetry class, but the marginal value of $b$ with respect to $c$ is greater than its marginal value with respect to $a$.
	
	We next show that SWS is not closed under convolution.
	The idea is to take a function $f$ over a set of items $M$ that has symmetries but is not SWS, then add an item $j$ that breaks all symmetries in $f$. The obtained function, call it $f_1$, is trivially SWS.
	Now let $f_2$ be a function that gives item $j$ a high value, and all other items 0 ($f_2$ is SWS), and consider the convolution over $f_1,f_2$; call it $g$.
	Given a set $S$, $g$ will assign $j$ to $f_2$ and all other items to $f_1$, so its value on subsets of $M$ coincides with $f$, which is not SWS. 
	For concreteness, let $f$ be the budget additive function over items $a,b,c$ with respective values $2,1,1$ and budget $2$ (see the average function above), and let $f_1$ be a function over $a,b,c,d$, such that $f_1(S)=f(S)$ for all $S \subseteq \{a,b,c\}$, and the marginal value of $d$ with respect to all subsets is 0, except for its marginal value with respect to $b$, which is $1$ (thus, the symmetry between $b,c$ breaks). Let $f_2$ be an additive function with value $100$ for $d$ and 0 for all other items. 
\end{proof}



\section{Proof of Lemma~\ref{lem:approach1}}
\label{sec:advanced-lemma}
\begin{proof} (of Lemma~\ref{lem:approach1})
	Color the sequence $\{0, \ldots, T\}$ periodically by $\log m$ colors $\{0, \ldots, \log m - 1\}$, where $t$ belongs to color class $C_i$ if $t = i$ modulo $\log m$. Decompose $H$ into $\log m$ collections $H^i$ (for $0 \le i < \log m$), where collection $H^i$ has those functions $h_t$ with $t\in C_i$. Let $f^i$ be the function satisfying $f^i(S) = \sum_{t\in C_i} h_t(S)$. Then we have the following sandwich property for every set $S$:
	
	$$\alpha \cdot \max_{0 \le i < \log m} f^i(S)  \le f(S) \le \beta \cdot \sum_{0 \le i < \log m} f^i(S)$$
	
	By monotonicity and concavity of $g$, {and using the assumptions that $0 \le \alpha \le 1$ and $\beta \ge 1$, we have}:
	
	$$\alpha \cdot \max_{0 \le i < \log m} g(f^i(S))  \le g(f(S)) \le \beta \cdot \sum_{0 \le i < \log m} g(f^i(S))$$
	
	We shall need the following claim:
	
	\begin{claim}
		\label{claim:constantGS}
		For every $i$, the function $g(f^i)$ can be approximated by a GS function $h^i$ within a constant factor.
	\end{claim}
	
	Claim~\ref{claim:constantGS} implies that for some constant {$c\ge 1$} and GS functions $h^i$, we have that:
	
	$$\alpha \max_{0 \le i < \log m} h^i(S)  \le g(f(S)) \le c\beta  \sum_{0 \le i < \log m} h^i(S)$$
	
	Using the above sandwich property, Lemma~\ref{lem:approach1} is an immediate corollary of Lemma~\ref{lem:approach}.

	It remains to prove Claim~\ref{claim:constantGS}.
	
	Consider the function $f^i = \sum_{t\in C_i} h_t$. (The functions $h_t$ are as defined in Lemma~\ref{lem:approach1}, not to be confused with the functions $h^i$ of Claim~\ref{claim:constantGS}.) Observe that for every set $S \subset M$ and for every $t,t' \in C_i$, if $t < t'$ and $h_{t'}(S) > 0$, then $h_t(S) \le m2^t \le 2^{t'} \le h_{t'}(S)$.
	
	Let $\tau: M \rightarrow (N \cup \bot)$ be a function that maps every item $x \in M$ to the highest value of $t \in C_i$ for which $x$ is in the support of $h_t$, and to $\bot$ if there is no such $t$. Let $\tau^{-1}(t) = \{x | x \in M \; , \; \tau(x) = t\}$. For every $t\in C_i$, let $h'_t$ be the function satisfying $h'_t(S) = h_t(S \cap \tau^{-1}(t))$ for every $S \subset M$. Then every item $x\in M$ is in the support of at most one of the functions $h'_t$. 
	
	Note that $\max_{t\in C_i} h_t(S) = \max_{t\in C_i} h'_t(S)$, because $\max_{t\in C_i} h_t(S)$ is attained at a $t^*$ that satisfies $t^* = \max_{x \in S} \tau(x)$, and $h'_{t^*}(S) = h_{t^*}(S)$. Moreover, $\sum_{t\in C_i} h'_t(S) \le f^i(S) \le 2\sum_{t\in C_i} h'_t(S)$ holds for every $S$, where the right inequality holds because the combined marginal values of items of $S$ dropped from all functions (with $t < t^*$) does not exceed $h_{t^*}(S) = h'_{t^*}(S)$.

	For every $t$ we use $\ell_t$ to denote $g(2^t)$ rounded down to the nearest power of~2, and $u_t$ to denote $g(2^t m)$ rounded down to the nearest power of~2. For $S \subset M$, if $|S \cap \tau^{-1}(t)| \ge 1$, then $\ell_t \le g(h'_t(S)) \le 2u_t$. For two consecutive members $t,t' \in C_i$ (hence $t' \ge t + \log m$), we have that $u_t \le \ell_{t'}$.
	
	We now define a {\em merging} operation. For $s > 2$, we refer to a sequence $t_1 < \ldots < t_s$ of indexes in color class $C_i$ as {\em mergeable} if it satisfies $\ell_{t_1} =  \ell_{t_s}$. As long as a mergeable sequence exists, we pick a maximal mergeable sequence (that is not contained in any longer mergeable sequence) and merge the functions $h'_{t_1}, \ldots, h'_{t_{s-1}}$ into one new function $h'_{t_1, \ldots, t_{s-1}}$. The items in the support of this function are $\tau^{-1}(t_1, \ldots, t_{s-1}) = \bigcup_{j \le s-1} \tau^{-1}(t_j)$. We define $h'_{t_1, \ldots, t_{s-1}}(S) = 2^{t_1}$ for sets that contain an item from $\tau^{-1}(t_1, \ldots, t_{s-1})$, and~0 otherwise. Observe that for every set $S \subset M$:
	
	$$g(h'_{t_1, \ldots, t_{s-1}}(S)) \le g(\sum_{j \le s-1} h'_{t_j}(S)) \le 2g(h'_{t_1, \ldots, t_{s-1}}(S))$$ 
	
	\noindent The right inequality follows because $u_{t_{s-1}} \le \ell_{t_s} = \ell_{t_1}$. 
	We consequently have that $\ell_{t_1, \ldots, t_{s-1}} = u_{t_1, \ldots, t_{s-1}} = g(2^{t_1})$.
	
	Observe that the function $g(h'_{t_1, \ldots, t_{s-1}})$ is a GS function (as $h'_{t_1, \ldots, t_{s-1}}$ has value either~0 or $2^{t_1}$). Likewise, for every unmerged function $h'_t$, Lemma~\ref{lem:concave-mrf-gs} implies that the corresponding function $g(h'_t)$ is a GS function, as $h'_t$ is an MRF (scaled by $2^t$).
	
	Given a color class $C_i$, after performing all applicable merge operations (and without changing the order of the functions), rename the functions that remain in $C_i$ as $f_1^i, f_2^i, \ldots $. Observe that having performed the merge operations, for every $j \ge 1$ it holds that $u^i_j \le \ell^i_{j+1} < \ell^i_{j+3}$ (where $\ell^i_j$ and $u^i_j$ are the natural renaming for the notation $\ell_t$ and $u_t$).
	
	With each function $f^i_j$ we associate the function $h^i_j = g(f^i_j)$. As discussed above, every $h^i_j$ function is GS. Moreover, we have the following sandwich property for every set $S \subset M$:
	
	$$\max_j h^i_j(S) \le g(f^i(S)) \le 4\sum_j h^i_j(S)$$
	
	\noindent (The factor~4 is a product of a factor~2 that is paid for switching from $h_t$ to $h'_t$, and a factor~2 that is paid for the merging operation.)
	
	We claim that for every $S$ it holds that $\sum_j h^i_j(S) \le 6\max_j h^i_j(S)$ (the constant~6 can be improved). Recall that the supports of the function $h^i_j$ are disjoint (for different $j$), and let $M_j$ denote the support of function $h^i_j$. To prove the claim, let $k$ be the largest index for which $|S \cap M_k| \ge 1$. Then $\max_j h^i_j(S) = h^i_k(S) \ge \ell^i_k$. For every $k' < k$ we have that  $h^i_{k'}(S) \le u^i_{k'} \le \ell^i_k$. Moreover, for every $j$, $u^i_j \le \frac{1}{2}u^i_{j+3}$. Hence the $u^i_j$ values can be partitioned into three geometric series that each contributes at most $2\ell^i_k$ to $\sum_j h^i_j(S)$, proving the claim.
	
	Let $h^i_u$ denote the function satisfying $h^i_u(S) = \sum_j h^i_j(S)$ for every $S$. Observe that $h^i_u$ is a GS function, as it is a sum over GS functions that have disjoint supports. Following the above claim and using concavity of $g$, for every $S \subset M$ we have that:
	
	$$\frac{1}{6}h^i_u(S) \le g(f^i(S)) \le 4h^i_u(S)$$
	
	Taking $h^i = \frac{1}{6}h^i_u$ proves Claim~\ref{claim:constantGS}.
\end{proof}

\section{Approximation of Budget Additive Valuations by Coverage Valuations}
\label{sec:coverage-ba}

\begin{theorem}
	\label{th:coverage-ba}
	Every BA function can be approximated by a coverage function within a ratio of $\rho = \frac{e}{e-1}$.
\end{theorem}

\begin{proof}
	Let $f$ be a BA function over $n$ items, where item $i$ has value $v_i$, and the budget is $B$. We now create a coverage function $g$. Let $N$ be sufficiently large. We shall have $N$ elements. Each item $i$ covers $\frac{v_i N}{B}$ elements chosen at random. Every element has weight $\frac{\rho B}{N}$. Hence for every item $i$ we have $g(i) = \rho f(i)$. For every set we have $g(S) \le \rho f(S)$, because $g$ is subadditive and upper bounded by $\rho B$.
	
	It remains to show that $g(S) \ge f(S)$ for every $S$. This holds by definition for sets that contain only one item, and likewise, for every set that contains an item of value at least $B$. For the remaining sets, suppose first that $f(S) < B$. The expected value of $g(S)$ (expectation over the randomness of the construction) is:
	
	$$E[g(S)] = \frac{\rho B}{N} N (1 - \prod_{i\in S} (1 - \frac{v_i}{B})) \ge \rho B (1 - e^{-\sum_{i \in S} \frac{v_i}{B}}) = \rho B (1 - e^{-\frac{f(S)}{B}})$$
	
	Denoting $x = \frac{f(S)}{B}$ we need to show that $\rho(1 - e^{-x}) \ge x$ in the range $0 \le x < 1$.  This indeed holds for $\rho = \frac{e}{e-1}$ because we then have equality for $x \in \{0,1\}$, and the function $1 - e^{-x}$ is concave.
	
	If $f(S) \ge B$, the same argument as above works, by first reducing the value of one or more of the items in $S$ such that after that $f(S) = B$ even without the budget constraint (details omitted).
	
	Finally, we note that by taking $N$ sufficiently large, $g(S)$ becomes arbitrarily close to $E[g(S)]$ simultaneously for all $S$. As a lower bound derived $E[g(S)]$ has some slackness (the inequality in the derivation is strict because $S$ contains an item of positive value strictly less than $B$), there is a choice of $g$ that satisfies $g(S) \ge f(S)$ for every $S$.
\end{proof}

We complement Theorem~\ref{th:coverage-ba} by showing that a ratio better than $\frac{e+1}{e}$ cannot be guaranteed, even by {\em matroid rank sum} (MRS) functions~\cite{Shioura2012,BalkanskiL18} --- sum over matroid rank functions (see Definition~\ref{def:mrs}) --- which is a strict superclass of coverage functions (coverage is the class of matroid rank functions of rank 1).
Closing the gap between $\frac{e+1}{e} \simeq 1.368$ and $\frac{e}{e-1} \simeq 1.582$ remains open.

\begin{proposition}
	Consider the BA function $f$ with budget $B$, a set $X = \{x_1, \dots, x_n\}$ of $n$ items, each of value~1, and an additional item $y$ of value $B$.
	For large $n$ and $B$ (e.g., $B = \sqrt{n}$), $f$ cannot be approximated by MRS within a ratio better than $\rho = \frac{e+1}{e}$.
\end{proposition}


\begin{proof}
	Let $B' \ge B$ be the total weight of matroids containing $y$, and let $X'$ be the total weight that matroids that do not contain $y$ can contribute (when all is $X$ present). We require $X' + B' \le \rho B$.
	
	Given an arbitrary solution, we can symmetrize it (average over all permutations) so that all $x_i$ are treated in exactly the same way.
	
	Assume that all matroids that contain $y$ have rank~1. (The bounds that follow hold without change even without this assumption, because the access above rank~1 can go into $X'$). Let $w \le \rho$ be the weight of matroids that contain both $y$ and $x_1$.
	
	Consider now a random set of $B$ items. The fraction of uncovered weight of matroids that contain $y$ is roughly $B'(1 - \frac{w}{B'})^B$. Hence $X' \ge B - B'(1 - \frac{w}{B'})^B$. For this $X'$, we get two lower bounds on $\rho$. Considering the grand bundle we have that $B\rho \ge B' + X'$, whereas considering only $x_1$ we have that $\rho \ge w + \frac{X'}{B}$.

	If $w = 1$ and $B' = B$, then $X' = \frac{B}{e}$ and $\rho = 1 + \frac{1}{e}$.
\end{proof}

\end{document}